\documentclass[12pt,cls,onecolumn]{IEEEtran}
\usepackage{graphicx,amsmath,amssymb,epsfig, amsfonts, cite, latexsym, cuted, multicol, multirow, subfigure, stfloats, array, tabularx}
\usepackage{subeqnarray}
\usepackage{color}
\usepackage{setspace}
\usepackage{anysize}

\begin{document}

\title{Ad Hoc Networking With Cost-Effective Infrastructure: Generalized Capacity Scaling}
\author{\large Cheol Jeong, \emph{Member}, \emph{IEEE} and Won-Yong Shin, \emph{Member}, \emph{IEEE} \\
\thanks{This work was supported by Basic Science Research Program through the National Research Foundation of Korea (NRF) funded by the Ministry of Science, ICT \& Future Planning (MSIP) (2012R1A1A1044151). This paper was presented in part at the 2014 IEEE International Symposium on Information Theory, Honolulu, HI, June/July 2014.}
\thanks{C. Jeong is with the DMC R\&D Center, Samsung Electronics, Suwon 443-742, Republic of Korea (E-mail: cheol.jeong@ieee.org).}
\thanks{W.-Y. Shin is with the Department of Computer Science and
Engineering, Dankook University, Yongin 448-701, Republic of Korea
(E-mail: wyshin@dankook.ac.kr).}
        %\thanks{The authors are with the School of Electrical Engineering and Computer Science, Korea Advanced Institute of Science and Technology (KAIST), Daejeon 305-701, Korea (E-mail: wyshin@stein.kaist.ac.kr; sychung@ee.kaist.ac.kr; yohlee@ee.kaist.ac.kr).}
} \maketitle

%\title{On the Power-Delay Trade-off in Wireless Ad Hoc Networks With Fading}
%\author{\authorblockN{Won-Yong Shin, Sae-Young Chung, and Yong H.
%Lee}\\
%\authorblockA{School of EECS, Korea Advanced Institute of Science and Technology, \\Daejeon, Korea\\
%Email: wyshin@stein.kaist.ac.kr; \{sychung, yohlee\}@ee.kaist.ac.kr}
%} \maketitle

%\markboth{Submitted to IEEE Transactions on Information Theory}
%{Jeong and Shin: Ad Hoc Networking With Rate-Limited Infrastructure:
%Generalized Capacity Scaling}

%%%%%%%%%%%%%%%%%%%%%%%%%%%%%%%%%%%%%%%%%%%%%%%%%%%%%%%%%%%%%%%%%%%%%%%%%%%%%%%%%%%%%%%%%%%%%%%%%%%%%%%%%%%%%%%%%%%%%%%%%%%%%%%%%%%%%%%%%%%%%%%%%%%%%

\newtheorem{definition}{Definition}%[section]
\newtheorem{theorem}{Theorem}%[section]
\newtheorem{lemma}{Lemma}%[section]
\newtheorem{example}{Example}
\newtheorem{corollary}{Corollary}%[section]
\newtheorem{proposition}{Proposition}%[section]
\newtheorem{conjecture}{Conjecture}%[section]
\newtheorem{remark}{Remark}%[section]

\def \diag{\operatornamewithlimits{diag}}
\def \min{\operatornamewithlimits{min}}
\def \max{\operatornamewithlimits{max}}
\def \log{\operatorname{log}}
\def \max{\operatorname{max}}
\def \rank{\operatorname{rank}}
\def \out{\operatorname{out}}
\def \exp{\operatorname{exp}}
\def \arg{\operatorname{arg}}
\def \E{\operatorname{E}}
\def \tr{\operatorname{tr}}
\def \SNR{\operatorname{SNR}}
\def \dB{\operatorname{dB}}
\def \ln{\operatorname{ln}}

\def \bmat{ \begin{bmatrix} }
\def \emat{ \end{bmatrix} }

\def \be {\begin{eqnarray}}
\def \ee {\end{eqnarray}}
\def \ben {\begin{eqnarray*}}
\def \een {\end{eqnarray*}}

\begin{abstract}
Capacity scaling of a large hybrid network with unit node density,
consisting of $n$ wireless {\em ad hoc} nodes, base stations (BSs)
equipped with multiple antennas, and one remote central processor
(RCP), is analyzed when wired backhaul links between the BSs and the
RCP are {\em rate-limited}. We deal with a general scenario where
the number of BSs, the number of antennas at each BS, and the
backhaul link rate can scale at arbitrary rates relative to $n$
(i.e., we introduce three scaling parameters). We first derive the
minimum backhaul link rate required to achieve the same capacity
scaling law as in the infinite-capacity backhaul link case. Assuming
an arbitrary rate scaling of each backhaul link, a generalized
achievable throughput scaling law is then analyzed in the network
based on using one of pure multihop, hierarchical cooperation, and
two infrastructure-supported routing protocols, and moreover,
three-dimensional information-theoretic operating regimes are
explicitly identified according to the three scaling parameters. In
particular, we show the case where our network having a power
limitation is also fundamentally in the degrees-of-freedom- or {\em
infrastructure-limited} regime, or both. In addition, a generalized
cut-set upper bound under the network model is derived by cutting
not only the wireless connections but also the wired connections. It
is shown that our upper bound matches the achievable throughput
scaling even under realistic network conditions such that each
backhaul link rate scales slower than the aforementioned
minimum-required backhaul link rate.
\end{abstract}


\newpage

\section{Introduction}
% Scaling law
%  - Ad hoc network
%  - Hybrid network
%  - Hybrid network with infinite-capacity infrastructure

% Scaling law with finite-capacity infrastructure

% Main contributions
% - Achievable rate scaling
% - Cut set upper bounds

In the era of the internet of things (IoT), referred to as a world
of massive devices equipped with sensors connected to the internet,
machine-to-machine (M2M) communications play an important role as an
emerging technology. As a large number of M2M devices participate in
communications, the protocol design efficiently delivering a number
of packets becomes more crucial. While numerical results via
computer simulations depend heavily on specific operating parameters
for a given protocol, a study on the capacity scaling of a
large-scale network provides a fundamental limit on the network
throughput and an asymptotic trend with respect to the number of
nodes. Hence, one can obtain insights on the practical design of a
protocol by characterizing the capacity scaling.

\subsection{Previous Work}

Gupta and Kumar's pioneering work~\cite{GuptaKumar:00} introduced
and characterized the sum throughput scaling law in a large wireless
{\it ad hoc} network. For the network having $n$ nodes randomly
distributed in a unit area, it was shown in~\cite{GuptaKumar:00}
that the total throughput scales as $\Theta(\sqrt{n/\log
n})$.\footnote{We use the following notation: i) $f(x)=O(g(x))$
means that there exist constants $C$ and $c$ such that $f(x)\leq
Cg(x)$ for all $x>c$, ii) $f(x)=o(g(x))$ means that
$\lim_{x\rightarrow \infty}\frac{f(x)}{g(x)}=0$, iii)
$f(x)=\Omega(g(x))$ if $g(x)=O(f(x))$, iv) $f(x)=w(g(x))$ if
$g(x)=o(f(x))$, and v) $f(x)=\Theta(g(x))$ if $f(x)=O(g(x))$ and
$g(x)=O(f(x))$ \cite{D.Knuth:76}.} This throughput scaling is
achieved by the nearest-neighbor multihop (MH) routing strategy.
In~\cite{FranceschettiDouseTseThiran:07,ShinChungLee:TIT13,ElGamalMammenPrabhakarShah:06},
MH schemes were further developed and analyzed in the network, while
their average throughput per source--destination (S--D) pair scales
far slower than $\Theta(1)$---the total throughput scaling was
improved to $\Theta(\sqrt{n})$ by using percolation
theory~\cite{FranceschettiDouseTseThiran:07}; the effect of
multipath fading channels on the throughput scaling was studied
in~\cite{ShinChungLee:TIT13}; and the trade-off between power and
delay was examined in terms of scaling laws
in~\cite{ElGamalMammenPrabhakarShah:06}. Together with the studies
on MH, it was shown that a hierarchical cooperation (HC)
strategy~\cite{OzgurLevequeTse:07,NiesenGuptaShah:10} achieves an
almost linear throughput scaling, i.e., $\Theta(n^{1-\epsilon})$ for
an arbitrarily small $\epsilon>0$, in the dense network of unit
area. As alternative approaches to achieving a linear scaling, novel
techniques such as networks with node
mobility~\cite{GrossglauserTse:02}, interference
alignment~\cite{CadambeJafar:08,Niesen:IT11}, directional
antennas~\cite{LiZhangFang:TMC11}, and infrastructure
support~\cite{ZemlianovVeciana:05} have also been proposed. Capacity
scaling was studied when the node density over the area is
inhomogeneous~\cite{GarettoGiacconeLeonardi:TN09-1,GarettoGiacconeLeonardi:TN09-2,AlfanoGarettoLeonardiMartina:TN10}.
In~\cite{YinGaoCui:TN10,HuangWang:TN12}, a cognitive network
consisting of primary and secondary networks was considered for
studying the scaling law.

Especially, since long delay and high cost of channel estimation are
needed in ad hoc networks with only wireless connectivity, the
interest in study of more amenable networks using infrastructure
support has greatly been growing. Such hybrid networks consisting of
both wireless ad hoc nodes and infrastructure nodes, or equivalently
base stations (BSs), have been introduced and analyzed
in~\cite{O.Dousse:INFOCOM02,KozatTassiulas:03,ZemlianovVeciana:05,LiuThiranTowsley:07,ShinJeonDevroyeVuChungLeeTarokh:08}.
Using very high frequency bands can be thought of as one of
promising ways to meet the high throughput requirements for the next
generation communications since abundant frequency resources are
available in these bands. A very small wavelength due to the very
high frequency bands enables us to deploy a vast number of antennas
at each BS (i.e., large-scale multi-antenna
systems~\cite{GuthyUtschickHonig:JSAC13}). In a hybrid network where
each BS is equipped with a large number of antennas, the optimal
capacity scaling was characterized by introducing two new routing
protocols, termed infrastructure-supported single-hop (ISH) and
infrastructure-supported multihop (IMH)
protocols~\cite{ShinJeonDevroyeVuChungLeeTarokh:08}. In the ISH
protocol, all wireless source nodes in each cell communicate with
its belonging BS using either a single-hop multiple-access or a
single-hop broadcast. In the IMH protocol, source nodes in each cell
communicate with its belonging BS via the nearest-neighbor MH
routing.

In hybrid networks with ideal
infrastructure~\cite{O.Dousse:INFOCOM02,KozatTassiulas:03,ZemlianovVeciana:05,LiuThiranTowsley:07,ShinJeonDevroyeVuChungLeeTarokh:08},
BSs have been assumed to be interconnected by infinite-capacity
wired links. In large-scale ad hoc networks, it is not
cost-effective to assume a long-distance optical fiber for all
BS-to-BS backhaul links. In practice, it is rather meaningful to
consider a cost-effective finite-rate backhaul link between BSs. One
natural question is what are the fundamental capabilities of hybrid
networks with {\it rate-limited} backhaul links in supporting $n$
nodes that wish to communicate concurrently with each other. To in
part answer this question, the throughput scaling was studied
in~\cite{C.Capar:11,C.Capar:12} for a simplified hybrid network,
where BSs are connected only to their neighboring BSs via a
finite-rate backhaul link---lower and upper bounds on the throughput
were derived in one- and two-dimensional networks. However,
in~\cite{C.Capar:11,C.Capar:12}, the system model under
consideration is comparatively simplified, and the form of
achievable schemes is limited only to MH routings.
In~\cite{JeongShin:ISIT13}, a general hybrid network deploying
multi-antenna BSs was studied in fundamentally analyzing how much
rate per BS-to-BS link is required to guarantee the optimal capacity
scaling achieved for the infinite-capacity backhaul link
scenario~\cite{ShinJeonDevroyeVuChungLeeTarokh:08}.

More practically, packets arrived at a certain BS in a radio access
network (RAN) are delivered to a core network (CN), and then are
transmitted from the CN to other BSs in the RAN, while neighboring
BSs have an interface through which only signaling information is
exchanged between them~\cite{S.Sesia:11}. The cellular network
operating based on a remote central processor (RCP) to which all BSs
are connected is well suited to this realistic
scenario~\cite{A.Sanderovich:ISIT07,P.Marsch:ICC07,S.Shamai:PIMRC08,B.Nazer:ISIT09,A.Sanderovich:TIT09}.
In~\cite{A.Sanderovich:TIT09}, the set of BSs connected to the RCP
via limited-capacity backhaul links was adopted in studying the
performance of the multi-cell processing in cooperative cellular
systems using Wyner-type models, where an achievable rate for the
uplink channel of such a cellular model was analyzed. To the best of
our knowledge, characterizing an information-theoretic capacity
scaling law of large hybrid networks (i.e. more general than the
Wyner-type model) with finite-capacity backhaul links in the
presence of the RCP has never been conducted before in the
literature.

\subsection{Contribution}

In this paper, we introduce a more general hybrid network with unit
node density (i.e., a hybrid extended network), consisting of $n$
wireless ad hoc nodes, multiple BSs equipped with multiple antennas,
and one RCP, in which wired backhaul links between the BSs and the
RCP are {\em rate-limited}. Our network model is well-suited for the
cloud-RAN that has recently received a great deal of attention by
assuming that some functionalities of the BSs are moved to a central
unit~\cite{ParkSimeoneSahinShamai:TSP13}. We take into account a
general scenario where three scaling parameters of importance
including i) the number of BSs, ii) the number of antennas at each
BS, and iii) each backhaul link rate can scale at arbitrary rates
relative to $n$. We first derive the minimum rate of a BS-to-RCP
link (or equivalently, an RCP-to-BS link) required to achieve the
same capacity scaling law as in the hybrid network with
infinite-capacity infrastructure. By showing that hybrid networks
can work optimally even with a cost-effective finite-rate backhaul
link (but not with the infinite-capacity backhaul link), we can
provide a vital guideline to design a {\em wireless} backhaul that
is an important component constituting future fifth generation (5G)
networks. Assuming an arbitrary rate scaling of each backhaul link,
we then analyze a new achievable throughput scaling law. Inspired by
the achievability result
in~\cite{ShinJeonDevroyeVuChungLeeTarokh:08}, for our achievable
scheme, we use one of pure MH, HC, and two different
infrastructure-supported routing protocols. Moreover, we identify
{\em three-dimensional} information-theoretic operating regimes
explicitly according to the aforementioned three scaling parameters.
In each operating regime, the best achievable scheme and its
throughput scaling results are shown. Besides the fact that extended
networks of unit node density are fundamentally
power-limited~\cite{OzgurJohariTseLeveque:10}, we are interested in
further finding the case for which our network, having a power
limitation, is either in the degrees-of-freedom (DoF)-limited
regime, where the performance is limited by the number of BSs or the
number of antennas per BS, or in the {\em infrastructure-limited}
regime, where the performance is limited by the rate of backhaul
links, or in both. Especially, studying the infrastructure-limited
regime would be very challenging when we want to efficiently design
the wireless backhaul given the number of BSs and the number of
antennas per BS. The infrastructure-limited regime in hybrid
networks has never been identified before in the literature. We thus
characterize these qualitatively different regimes according to the
three scaling parameters.

In addition, a generalized upper bound on the capacity scaling is
derived for our hybrid network with finite-capacity infrastructure
based on the cut-set theorem. In order to obtain a tight upper bound
on the aggregate capacity, we consider two different cuts under the
network model. The first cut divides the network area into two
halves by cutting the wireless connections. An interesting case is
the use of a new cut (i.e., the second cut), which divides the
network area into another halves by cutting not only the wireless
connections but also the wired connections. It is shown that our
upper bound matches the achievable throughput scaling for an
arbitrary rate scaling of the backhaul link rate. This indicates
that using one of the four routing protocols (i.e., pure MH, HC,
ISH, and IMH protocols) can achieve the optimal capacity scaling
even in the extended hybrid network with {\em rate-limited}
infrastructure, while the fundamental operating regimes are
identified accordingly depending on the backhaul link rate. Hence,
it turns out to be order-optimal for all three-dimensional operating
regimes of our network.

Our main contribution can be summarized as follows.
\begin{itemize}
\item The derivation of the minimum backhaul link rate required to achieve the same capacity scaling law as in the infinite-capacity backhaul link case.
\item The analysis of a generalized achievable throughput scaling law assuming an arbitrary rate scaling of each backhaul link.
\item The explicit identification of three-dimensional operating regimes according to the number of BSs, the number of antennas at each BS, and the backhaul link rate; and the characterization of both DoF- and infrastructure-limited regimes for our hybrid network which is fundamentally power-limited.
\item The derivation of a generalized cut-set upper bound for an arbitrary rate scaling of each backhaul link, resulting in the order optimality of the assumed network.
\end{itemize}

\subsection{Organization}

The rest of this paper is organized as follows. The system and
channel models are described in Section~\ref{SEC:System}. In
Section~\ref{SEC:Review}, the routing protocols with and without
infrastructure support are presented and their transmission rates
are shown. In Section~\ref{SEC:Routing}, the minimum required
backhaul link rate is derived, and then a generalized achievable
throughput scaling is analyzed. A generalized cut-set upper bound on
the capacity scaling is derived in
Section~\ref{SEC:CutSetUpperBound}. Finally,
Section~\ref{SEC:Conclusion} summarizes our paper with some
concluding remarks.

\subsection{Notations}

Throughout this paper, bold upper and lower case letters denote
matrices and vectors, respectively. The superscripts $T$ and
$\dagger$ denote the transpose and conjugate transpose,
respectively, of a matrix (or a vector). The matrix $\mathbf{I}_N$
is an $N\times N$ identity matrix, $[\cdot]_{ki}$ is the $(k,i)$th
element of a matrix, and $E[\cdot]$ is the expectation. The positive
semidefinite matrix $\mathbf{A}$ is denoted by $\mathbf{A}\geq 0$.

%In \cite{ZhangXuHanzo:10}, interference mitigation techniques are
%investigated under Wyner cellular model in relay-aided networks,
%where the relay is connected to the base station with a rate-limited
%optical fiber.

%The HC scheme, however, was not considered in the paper and the
%result is difficult to be extended to a general hybrid network in
%which each BS has multiple antennas.

%The rest of this paper is organized as follows.

\section{System and Channel Models} \label{SEC:System}
In an extended network of unit node density, $n$ nodes are uniformly
and independently distributed on a square of area $n$, except for
the area where BSs are placed.\footnote{It was shown that the HC
scheme is order-optimal in a dense network with
infrastructure~\cite{ShinJeonDevroyeVuChungLeeTarokh:08}. In other
words, the use of infrastructure is not required in a dense network.
Hence, the dense network model is out of interest in our paper where
the rate-limited infrastructure is considered.} We randomly pick
S--D pairings, so that each node acts as a source and has exactly
one corresponding destination node. Assume that the BSs are neither
sources nor destinations. As illustrated in
Fig.~\ref{Fig:RCPnetwork}, the network is divided into $m$ square
cells of equal area, where a BS with $l$ antennas is located at the
center of each cell. The total number of BS antennas in the network
is assumed to scale at most linearly with $n$, i.e., $ml=O(n)$. Note
that this assumption about the antenna scaling is feasible due to
the massive multiple-input multiple-output (MIMO) (or large-scale
MIMO) technology where each BS is equipped with a very large number
of antennas, which has recently received a lot of attention. For
analytical convenience, the parameters $n$, $m$, and $l$ are related
according to
\begin{align}
 n=m^{1/\beta}=l^{1/\gamma}, \nonumber
\end{align}
where $\beta,\gamma\in[0,1)$ with a constraint $\beta+\gamma \leq
1$. This constraint is reasonable since the total number of antennas
of all BSs in the network does not need to be larger than the number
of nodes in the network.

As depicted in Fig.~\ref{Fig:RCPnetwork}, it is assumed that all the
BSs are fully interconnected by wired links through one RCP.  For
simplicity, it is assumed that the RCP is located at the center of
the network. The packets transmitted from BSs are received at the
RCP and are then conveyed to the corresponding BSs. In the previous
works~\cite{O.Dousse:INFOCOM02,KozatTassiulas:03,ZemlianovVeciana:05,LiuThiranTowsley:07,ShinJeonDevroyeVuChungLeeTarokh:08},
the rate of backhaul links has been assumed to be unlimited so that
the links are not a bottleneck when packets are transmitted from one
cell to another. In practice, however, it is natural for each
backhaul link to have a finite capacity that may limit the
transmission rate of infrastructure-supported routing protocols. In
this paper, we assume that each BS is connected to one RCP through
an errorless wired link with {\em finite rate}
$R_{\textrm{BS}}=n^{\eta}$ for $\eta\in (-\infty,\infty)$. It is
also assumed that the BS-to-RCP or RCP-to-BS link is not affected by
interference.

\begin{figure}[t!]
  \centering
  \leavevmode \epsfxsize=4.3in
  \epsffile{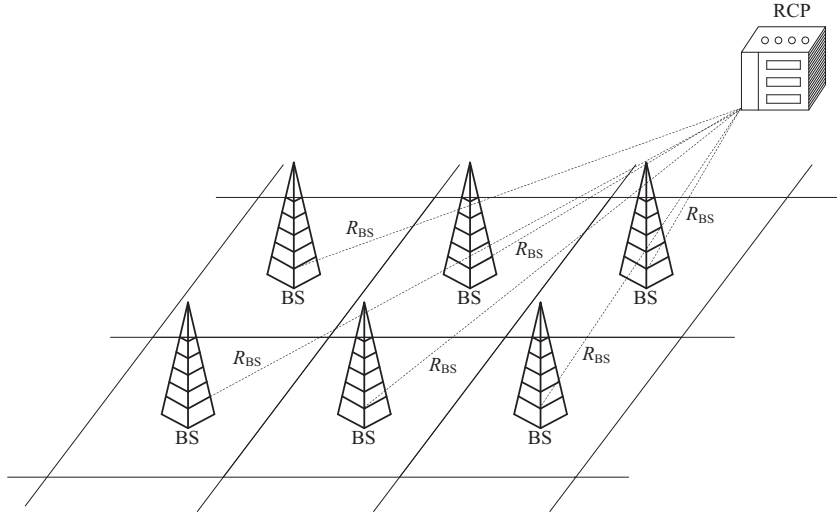}
  \caption{The hybrid network with limited backhaul link rate $R_{\textrm{BS}}$ between a BS and an RCP.}
  \label{Fig:RCPnetwork}
\end{figure}

The uplink channel vector between node $i$ and BS $b$ is denoted by
\begin{align} \label{EQ:uplinkCH}
    \mathbf{h}_{bi}^{(u)}=\left[\frac{e^{j\theta_{bi,1}^{(u)}}}{r_{bi,1}^{(u)\alpha/2}},
\frac{e^{j\theta_{bi,2}^{(u)}}}{r_{bi,2}^{(u)\alpha/2}},\ldots,
\frac{e^{j\theta_{bi,l}^{(u)}}}{r_{bi,l}^{(u)\alpha/2}}\right]^T,
\end{align}
where $\theta_{bi,t}^{(u)}$ represents the random phases uniformly
distributed over $[0,2\pi)$ based on a far-field assumption, which
is valid if the wavelength is sufficiently
small~\cite{FranceschettiDouseTseThiran:07,ShinJeonDevroyeVuChungLeeTarokh:08}.
Here, $r_{bi,t}^{(u)}$ denotes the distance between node $i$ and the
$t$th antenna of BS $b$, and $\alpha>2$ denotes the path-loss
exponent. The downlink channel vector between BS $b$ and node $i$ is
similarly denoted by
\begin{align} \label{EQ:downlinkCH}
\mathbf{h}_{ib}^{(d)}=\left[\frac{e^{j\theta_{ib,1}^{(d)}}}{r_{ib,1}^{(d)\alpha/2}},
\frac{e^{j\theta_{ib,2}^{(d)}}}{r_{ib,2}^{(d)\alpha/2}},\ldots,
\frac{e^{j\theta_{ib,l}^{(d)}}}{r_{ib,l}^{(d)\alpha/2}}\right].
\end{align}
The channel between nodes $i$ and $k$ is given by
\begin{align} \label{EQ:nodeCH}
h_{ki}=\frac{e^{j\theta_{ki}}}{r_{ki}^{\alpha/2}}.
\end{align}
%is similarly denoted by
%$\mathbf{h}_{ib}^{(d)}=\left[\frac{e^{j\theta_{ib,1}^{(d)}}}{r_{ib,1}^{(d)\alpha/2}},
%\frac{e^{j\theta_{ib,2}^{(d)}}}{r_{ib,2}^{(d)\alpha/2}},\ldots,
%\frac{e^{j\theta_{ib,l}^{(d)}}}{r_{ib,l}^{(d)\alpha/2}}\right]^T$.
%The channel between nodes $i$ and $k$ is given by
%$h_{ki}=\frac{e^{j\theta_{ki}}}{r_{ki}^{\alpha/2}}$.

%The received signal vector of BS $b$ is given by
%\begin{align}
%    \mathbf{y}_b =\sum_{i\in \mathcal{I}}^{}\mathbf{h}_{bi}x_i + \mathbf{n}_b,
%\end{align}
%where $x_i$ is the transmit signal of node $i$, $\mathcal{I}\subseteq \{1,\ldots,n\}$ is the set of simultaneously transmitting nodes, $\mathbf{h}_{bi}=\left[\frac{e^{j\theta_{bi,1}}}{r_{bi,1}^{\alpha/2}}, \frac{e^{j\theta_{bi,2}}}{r_{bi,2}^{\alpha/2}},\ldots, \frac{e^{j\theta_{bi,l}}}{r_{bi,l}^{\alpha/2}}\right]^T$ is the channel vector between node $i$ and BS $b$, and $\mathbf{n}_b$ is the circularly symmetric complex additive white Gaussian noise (AWGN) vector whose element has zero-mean and variance $N_0$.

For the uplink-downlink balance, it is assumed that each BS
satisfies an average transmit power constraint $nP/m$, while each
node satisfies an average transmit power constraint $P$. Then, the
total transmit power of all BSs is the same as the total transmit
power consumed by all wireless nodes. This assumption is based on
the same argument as duality connection between multiple access
channel (MAC) and broadcast channel (BC) in~\cite{ViswanathTse:03}.
The antenna configuration basically follows that
of~\cite{ShinJeonDevroyeVuChungLeeTarokh:08}. It is assumed that the
antennas of a BS are placed as follows:
\begin{enumerate}
  \item If $l=w(\sqrt{n/m})$ and $l=O(n/m)$, then $\sqrt{n/m}$ antennas are regularly placed on the BS boundary and the remaining antennas are uniformly placed inside the boundary.
  \item If $l=O(\sqrt{n/m})$, then $l$ antennas are regularly placed on the BS boundary.
\end{enumerate}
Such an antenna deployment guarantees both the nearest neighbor
transmission around the BS boundary and the enough spacing between
the antennas of each BS. This antenna configuration was adopted
in~\cite{JeongShin:ISIT13,GomezRanganErkip:ISIT14}. Let $R_n$ denote
the average transmission rate of each source. The total throughput
of the network is then defined as
$T_n(\alpha,\beta,\gamma,\eta)=nR_n$ and its scaling exponent is
given by\footnote{To simplify notations,
$T_n(\alpha,\beta,\gamma,\eta)$ and $e(\alpha,\beta,\gamma,\eta)$
will be written as $T_n$ and $e$, respectively, if dropping
$\alpha$, $\beta$, $\gamma$, and $\eta$ does not cause any
confusion.}
\begin{align}
    e(\alpha,\beta,\gamma,\eta)=\lim_{n\rightarrow\infty}\frac{\log T_n(\alpha,\beta,\gamma,\eta)}{\log n}. \nonumber
\end{align}

%%%%%%%%%%%%%%%%%%%%%%%%%%%%%%%%%%%%%%%%%%%%%%%%%%%%%%%%%%%%%%%%%%%%%%%%%%%%%%%%%%%%%%%%%%%%%%%%%%%%%%%%%%%%%%%%%%%%%%%%%%%%%%%%%%%%%%%%%%%%%%%%%%%%%
\section{Routing Protocols With and Without Infrastructure Support} \label{SEC:Review}
In this section, routing protocols with and without infrastructure
support are illuminated by describing each protocol in detail and
showing its achievable transmission rates in each cell.

\subsection{Routing Protocols With Infrastructure Support}

The routing protocols supported by BSs having multiple antennas
in~\cite{ShinJeonDevroyeVuChungLeeTarokh:08} are described with some
modification. In infrastructure-supported routing protocols, the
packet of a source is delivered to the corresponding destination of
the source using three stages: {\em access routing}, {\em backhaul
transmission}, and {\em exit routing}. In the access routing, the
packet of a source is transmitted to the home-cell BS. The packet
decoded at the home-cell BS is then transmitted to the target-cell
BS that is the nearest to the destination of the source via backhaul
links (i.e., both BS-to-RCP and RCP-to-BS links). In the exit
routing, the target-cell BS transmits the received packet to the
destination of the source. Let us start from the following lemma,
which quantifies the number of nodes in each cell.

\begin{lemma}\label{Lem:NumberNodesCell}
For $m<n$, the number of nodes in each cell is between
\begin{align}
    \left((1-\delta_0)\frac{n}{m},(1+\delta_0)\frac{n}{m}\right),
    \nonumber
\end{align}
i.e., $\Theta(n/m)$, with probability larger than $1-m
e^{-\Delta(\delta_0)n/m}$, where $\Delta(\delta_0)=(1+\delta_0)\ln
(1+\delta_0) -\delta_0$ for $0<\delta_0<1$ independent of $n$.
\end{lemma}

This lemma can be proved by slightly modifying the proof
of~\cite[Lemma 4.1]{OzgurLevequeTse:07}. According to the
transmission scheme in access and exit routings, the
infrastructure-supported routing protocols are categorized into two
different protocols as in the following.

%Suppose $m=n^\beta$ where $\beta\in[0,1)$. Then, the number of nodes inside each cell is between $((1-\delta_0)n^{1-\beta},(1+\delta_0)n^{1-\beta})$, i.e., $\Theta(n/m)$, with probability larger than $1-n^\beta e^{-\Delta(\delta_0)n^{1-\beta}}$ where $\Delta(\delta_0)=(1+\delta_0)\ln (1+\delta_0)$ for $0<\delta_0<1$ independent of $n$.

%In the access routing, the packet of a source is
%transmitted to the closest BS. The packet is then transmitted to the
%BS that is nearest to the destination of the source via BS-to-BS
%links. In the exit routing, the BS transmits it to the destination
%of the source.

\begin{figure}[t!]
  \centering
  \leavevmode \epsfxsize=5.in
  \epsffile{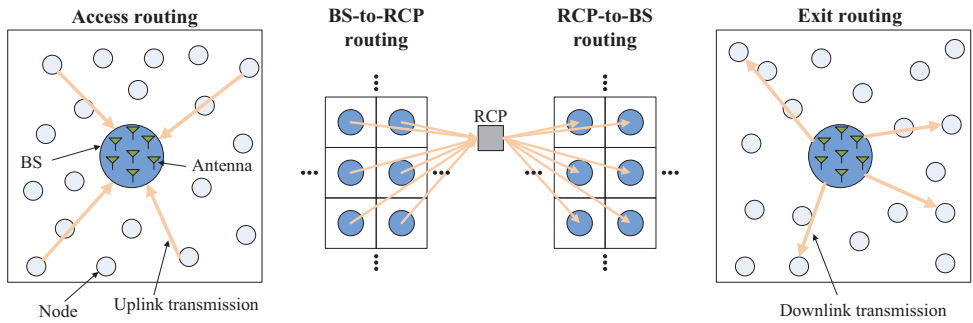}
  \caption{The ISH protocol. Each square represents a cell in the wireless network.}
  \label{Fig:ISHprotocol}
\end{figure}

\subsubsection{ISH Protocol}
There are $n/m$ nodes with high probability (whp) in each cell from
Lemma \ref{Lem:NumberNodesCell}. The ISH protocol is illustrated in
Fig. \ref{Fig:ISHprotocol} and each stage for the ISH protocol is
described as follows.
\begin{itemize}
  \item For the access routing, all source nodes in each cell transmit their
packets simultaneously to the home-cell BS via single-hop
multiple-access. A transmit power of $P$ is used at each node for
uplink transmission.
  \item The packets of source nodes are then jointly
decoded at the BS, assuming that the signals transmitted from the
other cells are treated as noise. Each BS performs a minimum
mean-square error (MMSE) estimation with successive interference
cancellation (SIC). More precisely, the $l\times 1$ unnormalized
receive filter $\mathbf{v}_i$ has the expression
\begin{align}
\mathbf{v}_i=\left(\mathbf{I}_l+\sum_{k>i}P\mathbf{h}_{sk}^{(u)}\mathbf{h}_{sk}^{(u)\dag}\right)^{-1}\mathbf{h}_{si}^{(u)},
\nonumber
\end{align}
which means that the receiver of BS $s$ for the $i$-th node cancels
signals from nodes $1,\cdots,i-1$ and treats signals from nodes
$i+1,\cdots,n/m$ as noise, for every $i$, when the canceling order
is given by $1,\cdots,n/m$.
  \item In the next stage, the decoded packets are transmitted from the BS to the RCP via BS-to-RCP link.
  \item The packets received at the RCP are conveyed to the corresponding BS via RCP-to-BS link.
  \item For the exit routing, each BS in each cell transmits $n/m$ packets received from the RCP, via single-hop broadcast to all the wireless nodes in its cell. The transmitters in the downlink are designed by the dual system of MMSE-SIC receive filters in the uplink, and thus perform an MMSE transmit precoding $\mathbf{u}_1,\cdots,\mathbf{u}_{n/m}$ with dirty paper coding at BS $s$:
\begin{align}
\mathbf{u}_i
=\left(\mathbf{I}_l+\sum_{k>i}p_k\mathbf{h}_{ks}^{(d)\dag}\mathbf{h}_{ks}^{(d)}\right)^{-1}\mathbf{h}_{is}^{(d)\dag},
\nonumber
\end{align}
where the power $p_k\geq 0$ allocated to each node satisfies
$\sum_{k}p_k\leq\frac{nP}{m}$ for $k=1,\cdots,n/m$. Note that a
total transmit power of $\frac{nP}{m}$ is used at each BS for
downlink transmission.
\end{itemize}

\begin{figure}[t!]
  \centering
  \leavevmode \epsfxsize=5.in
  \epsffile{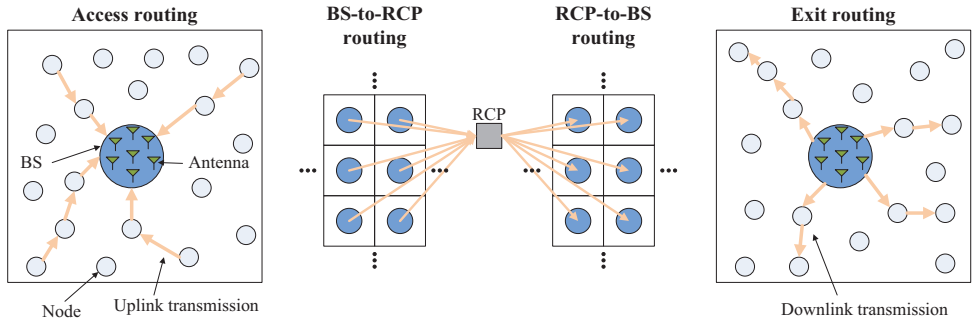}
  \caption{The IMH protocol. Each square represents a cell in the wireless network.}
  \label{Fig:IMHprotocol}
\end{figure}

\subsubsection{IMH Protocol}
Since the extended network is fundamentally power-limited, the ISH
protocol may not be effective especially when the node--BS distance
is quite long, which motivates us to introduce the IMH protocol. The
IMH protocol is illustrated in Fig. \ref{Fig:IMHprotocol} and each
stage for the IMH protocol is described as follows.
\begin{itemize}
  \item Each cell is
further divided into smaller square cells of area $2\log n$, termed
routing cells. Since $\min\{l,\sqrt{n/m}\}$ antennas are regularly
placed on the BS boundary, $\min\{l,\sqrt{n/m}\}$ MH paths can be
created simultaneously in each cell.
  \item For the access routing, the antennas placed only on the BS
boundary can receive the packets transmitted from one of the nodes
in the nearest-neighbor routing cell. Let us now consider how to set
an MH routing path from each source to the corresponding BS. Draw a
line connecting a source to one of the antennas of its BS and
perform MH routing horizontally or vertically by using the adjacent
routing cells passing through the line until its packets reach the
corresponding receiver (antenna).
  \item The BS-to-RCP and RCP-to-BS transmission is the
same as the ISH protocol case.
  \item For the exit routing, each antenna on the
BS boundary transmits the packets to one of the nodes in the
nearest-neighbor routing cell. Each antenna on the BS boundary
transmits its packets via MH transmissions along a line connecting
the antenna of its BS to the corresponding destination.
\end{itemize}

\subsection{Routing Protocols Without Infrastructure Support}
The protocols based only on infrastructure support may not be
sufficient to achieve the optimal capacity scaling especially when
$m$ and $l$ are small. Using one of the MH
transmission~\cite{GuptaKumar:00} and the HC
strategy~\cite{OzgurLevequeTse:07} may be beneficial in terms of
improving the achievable throughput scaling.
%In the following, the
%MH and HC protocols are briefly explained.

\subsubsection{MH Protocol}
The MH protocol is described as follows.
\begin{itemize}
  \item The network is divided into square routing cells of area $2\log n$.
  \item Draw a line connecting a source to its destination and
  perform MH routing horizontally or vertically by using the
  adjacent routing cells passing through the line until its packets
  reach the corresponding destination.
  \item A transmit power of $P$ is used.
  \item Each routing cell operates the $k$-time division multiple access scheme to avoid huge interference, where $k>0$ is some small constant independent of $n$.
\end{itemize}

\subsubsection{HC Protocol}
The procedure of the HC protocol is as follows.
\begin{itemize}
  \item The network is divided into clusters each having $M$ nodes.
  \item (Phase 1) Each source in a cluster transmits its packets to the other $M-1$ nodes in the same cluster.
  \item (Phase 2) A long-range MIMO transmission is performed between two clusters having a source and its destination.
  \item (Phase 3) Each node quantizes the received observations and delivers
  the quantized data to the rest of nodes in the same cluster. By
  collecting all quantized observations, each destination can decode
  its packets.
\end{itemize}

When each node transmits data within its cluster in Phases 1 and 3,
it is possible to apply another smaller-scaled cooperation within
each cluster by dividing each cluster into smaller ones. By
recursively applying this procedure, it is possible to establish the
HC strategy in the network.

\subsection{The Transmission Rates of Routing Protocols}\label{Sec:AchievableThroughputProtocols}

As addressed earlier, the RCP is incorporated into the hybrid
network model using BSs. In this subsection, we show how much
transmission rates are obtained via wireless links between ad hoc
nodes and home-cell BSs for each infrastructure-supported routing
protocol. We remark that the transmission rates in both access and
exit routings are irrelevant to the rate of backhaul links and thus
are essentially the same as the infinite-capacity backhaul link
case~\cite{ShinJeonDevroyeVuChungLeeTarokh:08}. The transmission
rates of each routing protocol are given in the following lemmas.

\begin{lemma}[\cite{ShinJeonDevroyeVuChungLeeTarokh:08}]\label{Lem:RateISH}
Suppose that the ISH protocol is used in the hybrid network of unit
node density. Then, the transmission rate in each cell for both
access and exit routings is given by
\begin{align}\label{Eq:ThroughputISH-BSunlimited}
    \Omega\left( l\left(\frac{m}{n}\right)^{\alpha/2-1}\right).
\end{align}
\end{lemma}

\begin{lemma}[\cite{ShinJeonDevroyeVuChungLeeTarokh:08}]\label{Lem:RateIMH}
Suppose that the IMH protocol is used in the hybrid network of unit
node density. Then, the transmission rate in each cell for both
access and exit routings is given by
\begin{align}\label{Eq:ThroughputIMH-BSunlimited}
    \Omega\left(\min\left\{
    l,\left(\frac{n}{m}\right)^{1/2-\epsilon}\right\}\right),
\end{align}
where $\epsilon>0$ is an arbitrarily small constant.
\end{lemma}

The total throughput scaling laws achieved by the MH and HC
protocols that utilize no infrastructure were derived in
\cite{GuptaKumar:00} and \cite{OzgurLevequeTse:07}, respectively,
and are given as follows:
\begin{align}
    T_{n,\textrm{MH}}&=\Omega(n^{1/2-\epsilon}) \nonumber
\end{align}
and
\begin{align}
 T_{n,\textrm{HC}}&=\Omega(n^{2-\alpha/2-\epsilon}), \nonumber
\end{align}
where $\epsilon>0$ is an arbitrarily small constant.

%%%%%%%%%%%%%%%%%%%%%%%%%%%%%%%%%%%%%%%%%%%%%%%%%%%%%%%%%%%%%%%%%%%%%%%%%%%%%%%%%%%%%%%%%%%%%%%%%%%%%%%%%%%%%%%%%%%%%%%%%%%%%%%%%%%%%%%%%%%%%%%%%%%%%
\section{Achievability Result} \label{SEC:Routing}

In this section, we first introduce information-theoretic
two-dimensional operating regimes with respect to the number of BSs
and the number of antennas per BS (i.e., the scaling parameters
$\beta$ and $\gamma$) for the infinite-capacity backhaul link
scenario. We then derive the minimum rate of each backhaul link,
required to achieve the optimal capacity scaling, according to the
two-dimensional operating regimes. Assuming that the rate of each
backhaul link scales at an arbitrary rate relative to $n$, we
characterize a new achievable throughput scaling, which generalizes
the existing achievability result
in~\cite{ShinJeonDevroyeVuChungLeeTarokh:08}. The
infrastructure-limited regime in which the throughput scaling is
limited by the rate of finite-capacity backhaul links is also
identified. Furthermore, we closely scrutinize our achievability
result according to the {\em three-dimensional operating regimes}
identified by introducing a new scaling parameter $\eta$.

\subsection{Two-Dimensional Operating Regimes With Infinite-Capacity Infrastructure}

The optimal capacity scaling was derived in
\cite{ShinJeonDevroyeVuChungLeeTarokh:08} for hybrid networks with
no RCP when the rate of each BS-to-BS link is unlimited. Although
our hybrid network characterized in the presence of RCP differs from
the network model in~\cite{ShinJeonDevroyeVuChungLeeTarokh:08}, the
existing analytical result, including the optimal capacity scaling
and information-theoretic operating regimes, can be
straightforwardly applied to our network setup for the
infinite-capacity backhaul link case, i.e., $\eta\rightarrow\infty$.

\begin{figure}[t!]
  \centering
  \leavevmode \epsfxsize=3.0in
  \epsffile{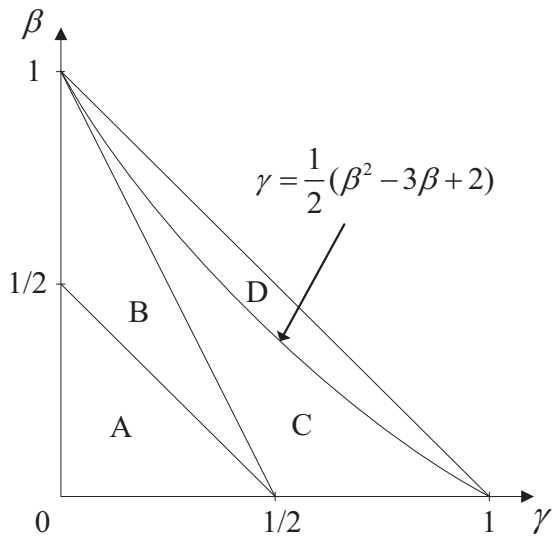}
  \caption{The operating regimes on the achievable throughput scaling with respect to $\beta$ and $\gamma$ for $\eta\rightarrow\infty$.}
  \label{Fig:OperatingRegimesInfiniteBScapacity}
\end{figure}

As illustrated in Fig.~\ref{Fig:OperatingRegimesInfiniteBScapacity},
when $\eta\rightarrow\infty$, two-dimensional operating regimes with
respect to $\beta$ and $\gamma$ are divided into four sub-regimes.
To be specific, the best strategy among the four schemes ISH, IMH,
MH, and HC depends on the path-loss exponent $\alpha$ and the two
scaling parameters $\beta$ and $\gamma$ under the network, and the
regimes at which the best achievable throughput is determined
according to the value of $\alpha$ synthetically constitute our
operating regimes. The best scheme and its corresponding scaling
exponent $e(\alpha,\beta,\gamma,\infty)$ in each regime are
summarized in Table~\ref{Tab:RateExtended}. The four sub-regimes are
described in more detail as follows.
\begin{itemize}
    \item In Regime A, the infrastructure is not helpful to improve the capacity scaling since $\beta$ and $\gamma$ are too small.
    \item In Regime B, the HC and IMH protocols are used to achieve the optimal capacity scaling. As $\alpha$ increases, the IMH protocol outperforms the HC since long-range MIMO transmissions of the HC becomes inefficient at the high path-loss attenuation regime.
    \item In Regime C, using the HC and IMH protocols guarantees the order optimality as in Regime B, but leads to a different throughput scaling from that in Regime B.
    \item In Regime D, the HC protocol has the highest throughput when $\alpha$ is small, but as $\alpha$ increases, the best scheme becomes the ISH protocol. Finally, the IMH protocol becomes dominant when $\alpha$ is very large since the ISH protocol has a power limitation at the high path-loss attenuation regime.
\end{itemize}

In the next subsections, we shall derive the minimum required rate
of backhaul links under each of these operating regimes. In
addition, by introducing the scaling parameter $\eta$, we shall
identify new three-dimensional operating regimes while
characterizing a generalized achievable throughput scaling for the
hybrid network with finite-capacity infrastructure.

\begin{table}[t]
    \centering
    \caption{Achievability result for a hybrid extended network with infinite-capacity infrastructure~\cite{ShinJeonDevroyeVuChungLeeTarokh:08}}
    \label{Tab:RateExtended}
\begin{tabular}{|c|c|c|c|}
  \hline
  Regime & Condition & Best scheme & Throughput scaling exponent $e$\\
  \noalign{\hrule height 1.2pt}
  \hline
  \multirow{2}{*}{A} & $2<\alpha<3$  & HC & $2-\frac{\alpha}{2}$\\
                     & $\alpha\geq 3$ & MH & $\frac{1}{2}$\\
  \hline
  \multirow{2}{*}{B} & $2<\alpha<4-2\beta-2\gamma$  & HC & $2-\frac{\alpha}{2}$\\
                     & $\alpha\geq 4-2\beta-2\gamma$ & IMH & $\beta+\gamma$\\
  \hline
  \multirow{2}{*}{C} & $2<\alpha<3-\beta$  & HC & $2-\frac{\alpha}{2}$ \\
                     & $\alpha\geq 3-\beta$ & IMH & $\frac{1+\beta}{2}$ \\
  \hline
  \multirow{3}{*}{D} & $2<\alpha<\frac{2(1-\gamma)}{\beta}$  & HC & $2-\frac{\alpha}{2}$\\
                     & $\frac{2(1-\gamma)}{\beta}\leq\alpha<1+\frac{2\gamma}{1-\beta}$ & ISH & $1+\gamma-\frac{\alpha(1-\beta)}{2}$\\
                     & $\alpha\geq 1+\frac{2\gamma}{1-\beta}$ & IMH & $\frac{1+\beta}{2}$\\
  \hline
\end{tabular}
\end{table}

\subsection{The Minimum Required Rate of Backhaul Links}

The supportable transmission rate of the backhaul link between BSs
and the RCP may increase proportionally with the cost that one needs
to pay. In order to give a cost-effective backhaul solution for a
large-scale network, we would like to derive the minimum rate
scaling of each backhaul link required to achieve the same capacity
scaling law as in the infinite-capacity backhaul link case.
According to the two-dimensional operating regimes in Fig.
\ref{Fig:OperatingRegimesInfiniteBScapacity}, the required rate of
each BS-to-RCP link (or each RCP-to-BS link), denoted by
$C_{\textrm{BS}}$, is derived in the following theorem.

\begin{theorem}\label{Thm:MinReqRate}
The minimum rate of each backhaul link required to achieve the
optimal capacity scaling of hybrid networks with infinite-capacity
infrastructure is given by
\begin{align}\label{Eq:MinReqRateCBS}
    C_{\textrm{BS}} = \left\{ \begin{array}{ll}
    0 &\textrm{  for Regime A}\\
    \Omega\left(l\right) &\textrm{  for Regime B}\\
    \Omega\left(\left(\frac{n}{m}\right)^{1/2-\epsilon}\right) &\textrm{  for Regime C}\\
    \Omega\left(l\left(\frac{m}{n}\right)^{\log_m (n/l)-1}\right) & ~\textrm{for Regime D} \\
    \end{array} \right.
\end{align}
for an arbitrarily small constant $\epsilon>0$. The associated
operating regimes with respect to $\beta$ and $\gamma$ are
illustrated in Fig.~\ref{Fig:OperatingRegimesInfiniteBScapacity}.
\end{theorem}

\begin{proof}
The required rate of each backhaul link is determined by the
multiplication of the number of S--D pairs that transmit packets
simultaneously through each link and the transmission rate of the
infrastructure-supported routing protocols for each S--D pair. From
Table~\ref{Tab:RateExtended}, no infrastructure-supported protocol
is needed in Regime A to achieve the optimal capacity scaling,
thereby resulting in $C_{\textrm{BS}}=0$ in the regime. Let us first
focus on the IMH protocol, which is used in Regimes B, C, and D when
$\alpha$ is greater than or equal to $4-2\beta-2\gamma$, $3-\beta$,
and $1+\frac{2\gamma}{1-\beta}$, respectively (see
Table~\ref{Tab:RateExtended}). Let $T_{n,\textrm{IMH}}$ denote the
aggregate throughput achieved by the IMH protocol when the rate of
each backhaul link is unlimited. Then, from Lemma~\ref{Lem:RateIMH},
it follow that $T_{n,\textrm{IMH}}=\Omega\left(m\min\left\{
l,\left(\frac{n}{m}\right)^{1/2-\epsilon}\right\}\right)$. Since
only $\min\{l,\sqrt{n/m}\}$ nodes among $n/m$ nodes in each cell
transmit packets using the IMH protocol, the transmission rate of
each activated S--D pair is given by
\begin{align}
    \frac{T_{n,\textrm{IMH}}}{n}\frac{n}{m}\frac{1}{\min\left\{l,\sqrt{\frac{n}{m}}\right\}}
    =\frac{T_{n,\textrm{IMH}}}{m}\frac{1}{\min\left\{l,\sqrt{\frac{n}{m}}\right\}}.
    \nonumber
\end{align}
Note that under the IMH protocol, the number of S--D pairs that
transmit packets simultaneously through each backhaul link is
$\min\left\{l,\sqrt{\frac{n}{m}}\right\}$. Hence, in Regimes B, C,
and D, the minimum required rate of each backhaul link to guarantee
the throughput $T_{n,\textrm{IMH}}$, denoted by
$C_{\textrm{BS,IMH}}$, is given by
\begin{align}\label{Eq:PerCellRequiredRateIMH}
    C_{\textrm{BS,IMH}}&=\frac{T_{n,\textrm{IMH}}}{m}\frac{1}{\min\left\{l,\sqrt{\frac{n}{m}}\right\}}
    \min\left\{l,\sqrt{\frac{n}{m}}\right\}
    \nonumber\\
    &=\Omega\left(\min\left\{l,\left(\frac{n}{m}\right)^{1/2-\epsilon}
    \right\}
    \right)
    \nonumber\\
    &=\left\{\begin{array}{ll}
    \Omega(l) &\textrm{  for Regime B}\\
    \Omega\left(\left(\frac{n}{m}\right)^{1/2-\epsilon}\right) &\textrm{  for Regimes C and D,}\\
    \end{array}
    \right.
\end{align}
which is the same as $C_{\textrm{BS}}$ in Regimes B and C since only
the IMH protocol is used between the two infrastructure-supported
protocols in these regimes. Now let us turn to the ISH protocol,
which is used in Regime D when
$\frac{2(1-\gamma)}{\beta}\leq\alpha<1+\frac{2\gamma}{1-\beta}$ (see
Table~\ref{Tab:RateExtended}). Let $T_{n,\textrm{ISH}}$ denote the
aggregate throughput achieved by the ISH protocol when the rate of
each backhaul link is unlimited. Then, from Lemma~\ref{Lem:RateISH},
it follows that
$T_{n,\textrm{ISH}}=\Omega\left(ml\left(\frac{m}{n}\right)^{\alpha/2-1}
\right)$. Since each S--D pair transmits at a rate
$T_{n,\textrm{ISH}}/n$ and the number of S--D pairs that transmit
packets simultaneously through each link is $n/m$, the minimum
required rate of each backhaul link for a given $\alpha$ to
guarantee the throughput $T_{n,\textrm{ISH}}$, denoted by
$C_{\textrm{BS,ISH}}$, is given by
\begin{align}\label{Eq:PerCellRequiredRateISH}
    C_{\textrm{BS,ISH}}&=\frac{T_{n,\textrm{ISH}}}{n}\frac{n}{m}=\frac{T_{n,\textrm{ISH}}}{m}
    =\Omega\left(l\left(\frac{m}{n}\right)^{\alpha/2-1}
    \right).
\end{align}

\begin{figure}[t!]
  \centering
  \leavevmode \epsfxsize=3.7in
  \epsffile{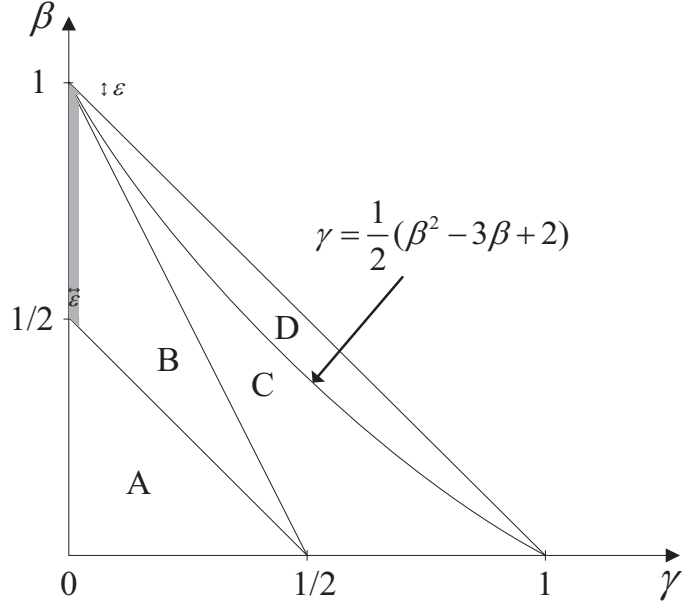}
  \caption{The operating regime in which $C_{\textrm{BS}}=O(n^\epsilon)$ for an arbitrarily small $\epsilon>0$ (represented by the shaded area).}
  \label{Fig:SmallScalingLaw}
\end{figure}

Since either the ISH or IMH protocol can be used according to the
value of $\alpha$ in Regime D, we should compare the required rates
of backhaul links for both protocols. Let us find the minimum
required rate $C_{\textrm{BS}}$ with which
$\max\{T_{n,\textrm{IMH}},T_{n,\textrm{ISH}}\}$ can be guaranteed
for all $\alpha$. Using (\ref{Eq:PerCellRequiredRateIMH}) and
(\ref{Eq:PerCellRequiredRateISH}), in Regime D, we have
\begin{align}
    &C_{\textrm{BS}}
    \nonumber\\
    &= \Omega \left(\max\left\{\left(\frac{n}{m}\right)^{1/2-\epsilon},
    \right.\right.
    \nonumber\\
    &\left.\left.
    ~~~~~~~~~~~~~~~\max_{\frac{2(1-\gamma)}{\beta}\leq\alpha<1+\frac{2\gamma}{1-\beta}}
    l\left(\frac{m}{n}\right)^{\alpha/2-1}\right\}\right)
    \nonumber\\
    &=\Omega \left(\max\left\{\left(\frac{n}{m}\right)^{1/2-\epsilon},
    l\left(\frac{m}{n}\right)^{(1-\gamma)/\beta-1}\right\}\right)
    \nonumber\\
    &=\Omega
    \left(l\left(\frac{m}{n}\right)^{(1-\gamma)/\beta-1}\right)\nonumber\\
    &=\Omega\left(l\left(\frac{m}{n}\right)^{\log_m (n/l)-1}\right),
    \nonumber
\end{align}
where the third equality holds since
$\gamma\geq\frac{1}{2}(\beta^2-3\beta+2)$ for Regime D. Therefore,
the minimum required rate of each backhaul link, $C_{\textrm{BS}}$,
is finally given by (\ref{Eq:MinReqRateCBS}), which completes the
proof of Theorem~\ref{Thm:MinReqRate}.
\end{proof}

This result indicates that a judicious rate scaling of the BS-to-RCP
link (or the RCP-to-BS link) under a given operating regime leads to
the order optimality of our general hybrid network along with
cost-effective backhaul links.

\begin{remark}[Negligibly small backhaul link rates]
It is obvious to see that the backhaul is not needed at all in
Regime A where using either the MH or HC protocol leads to the best
throughput performance of the network. An interesting observation is
now to find other regimes in which the minimum required rate of each
backhaul link is negligibly small, i.e.,
$C_{\textrm{BS}}=O(n^\epsilon)$ for an arbitrarily small
$\epsilon>0$. From Theorem \ref{Thm:MinReqRate}, it is shown that
$C_{\textrm{BS}}=O(n^\epsilon)$ if $\gamma=\epsilon$ in Regimes B
and D or if $\beta=1-\epsilon$ in Regimes C and D, which is depicted
in Fig. \ref{Fig:SmallScalingLaw}. This result reveals that for the
case where the number of antennas at each BS is very small or the
number of BSs is almost the same as the number of nodes, the
backhaul link rate $R_{\textrm{BS}}$ does not need to be infinitely
high even for a large number of wireless nodes in the network.
\end{remark}

%%%%%%%%%%%%%%%%%%%%%%%%%%%%%%%%%%%%%%%%%%%%%%%%%%%%%%%%%%%%%%%%%%%%%%%%%%%%%%%%%%%%%%%%%%%%%%%%%%%%%%%%%%%%%%%%%%%%%%%%%%%%%%%%%%%%%%%%%%%%%%%%%%%%%
\subsection{Generalized Achievable Throughput Scaling With Finite-Capacity Infrastructure}
If the rate of each backhaul link, $R_{\textrm{BS}}$, is greater
than or equal to the minimum required rate $C_{\textrm{BS}}$ in
Theorem~\ref{Thm:MinReqRate}, then the achievable throughput scaling
$T_n$ in the hybrid network with finite-capacity infrastructure is
the same as the infinite-capacity infrastructure backhaul link
scenario. Otherwise, $T_n$ will be decreased accordingly depending
on the operating regimes for which the infrastructure-supported
routing protocols are used. In this subsection, a generalized
achievable throughput scaling is derived with an arbitrary rate
scaling of each BS-to-RCP or RCP-to-BS link (or with the scaling
parameter $\eta\in(-\infty,\infty)$). The three-dimensional
operating regimes with respect to the number of BSs, $m$, the number
of antennas per BS, $l$, and the backhaul link rate,
$R_{\textrm{BS}}$, are also identified. We start from establishing
the following theorem.

\begin{theorem}\label{Thm:AchievableRateLimited}
In the hybrid network with the backhaul link rate $R_{\textrm{BS}}$,
the aggregate throughput $T_n$ scales as
\begin{align}
    &\Omega\left(
    \max\left\{\min\left\{
    \max\left\{ml\left(\frac{m}{n}\right)^{\alpha/2-1},
    \right.\right.\right.\right.
    \nonumber\\
    &\left.\left.
    ~~~~~~~~~~~~~~~~~~~~~~~~~~
    \min\left\{ ml,m\left(\frac{n}{m}\right)^{1/2-\epsilon}\right\}\right\}
    ,m R_{\textrm{BS}}\right\},
    \nonumber\\
    &\left.\left.
    ~~~~~~~~~~~n^{1/2-\epsilon},n^{2-\alpha/2-\epsilon}\right\}
    \right),
\end{align}
where $\epsilon>0$ is an arbitrarily small constant.
\end{theorem}

\begin{proof}
When the rate of each backhaul link is limited by $R_{\textrm{BS}}$,
from (\ref{Eq:ThroughputISH-BSunlimited}) and
(\ref{Eq:ThroughputIMH-BSunlimited}), the aggregate rates achieved
using the ISH and IMH protocols are given by
\begin{align}
    T_{n,\textrm{ISH}}=\Omega\left(\min\left\{ml\left(\frac{m}{n}\right)^{\alpha/2-1},m R_{\textrm{BS}}\right\}\right)
    \nonumber
\end{align}
and
\begin{align}
    T_{n,\textrm{IMH}}=\Omega \left( \min\left\{ ml,m\left(\frac{n}{m}\right)^{1/2-\epsilon},mR_{\textrm{BS}}\right\}
    \right), \nonumber
\end{align}
respectively, where $mR_{\textrm{BS}}$ represents the maximum
supportable rate of backhaul links. We then have
\begin{align}
    &\max\{T_{n,\textrm{ISH}},T_{n,\textrm{IMH}}\}
    \nonumber\\
    &=\Omega\left(\min\left\{
    \max\left\{ml\left(\frac{m}{n}\right)^{\alpha/2-1},
    \right.\right.\right.
    \nonumber\\
    &\left.\left.\left.
    ~~~~~~~~~~~~~~~~~~~~~~\min\left\{ ml,m\left(\frac{n}{m}\right)^{1/2-\epsilon}\right\}\right\}
    ,m R_{\textrm{BS}}\right\}\right) \nonumber
\end{align}
since $\max\{\min\{a,x\},\min\{b,x\}\}=\min\{\max\{a,b\},x\}$.
Finally, the achievable total throughput of the network is
determined by the maximum of the aggregate rates achieved by the
ISH, IMH, MH, and HC protocols, and thus is given by
\begin{align}
    T_{n}&=\max\left\{T_{n,\textrm{ISH}},T_{n,\textrm{IMH}},n^{1/2-\epsilon},n^{2-\alpha/2-\epsilon}\right\}
    \nonumber\\
    &=\Omega\left(
    \max\left\{\min\left\{
    \max\left\{ml\left(\frac{m}{n}\right)^{\alpha/2-1},
    \right.\right.\right.\right.
    \nonumber\\
    &\left.\left.
    ~~~~~~~~~~~~~~~~~~~~~~~~~~
    \min\left\{ ml,m\left(\frac{n}{m}\right)^{1/2-\epsilon}\right\}\right\}
    ,m R_{\textrm{BS}}\right\},
    \nonumber\\
    &\left.\left.
    ~~~~~~~~~~~~~~n^{1/2-\epsilon},n^{2-\alpha/2-\epsilon}\right\}
    \right),
    \nonumber
\end{align}
which completes the proof of
Theorem~\ref{Thm:AchievableRateLimited}.
\end{proof}

In the network with rate-limited infrastructure, either the MH and
HC protocol may outperform the infrastructure-supported protocols
even under certain operating regimes such that using either the ISH
or IMH protocol leads to a better throughput scaling for the
rate-unlimited infrastructure scenario. This is because the
throughput achieved by the ISH and IMH protocols can be severely
decreased when the rate $R_{\textrm{BS}}$ becomes the bottleneck.
Note that our hybrid extended network is fundamentally
power-limited~\cite{OzgurLevequeTse:07}. It is also worth noting
that the network may have either a DoF or an infrastructure
limitation, or both. In the DoF-limited regime, the performance is
limited by the number of BSs or the number of antennas per BS. On
the other hand, in the infrastructure-limited regime, the
performance is limited by the rate of backhaul links. In the
following two remarks, we show the case where our network has such
fundamental limitations (the term $\epsilon$ is omitted for
notational convenience).

\begin{remark}[DoF-limited regimes]\label{Rem:DoF-limitedRegimes}
As seen in Table~\ref{Tab:RateExtended}, in Regimes $\textrm{B}$,
$\textrm{C}$, and $\textrm{D}$, when $\alpha$ is greater than or
equal to a certain value, using the ISH or IMH protocol yields the
best throughput while the throughput scaling exponent depends on
$\beta$ or $\gamma$, or both. In other words, the performance is
limited by the number of BSs, $m$, or the number of antennas per BS,
$l$, or both. Thus, one can say that these high path-loss
attenuation regimes are {\it DoF-limited}.
\end{remark}

\begin{remark}[Infrastructure-limited regimes]\label{Rem:Infra-limitedRegimes}
Let us introduce the {\it infrastructure-limited} regime where the
performance is limited by the backhaul link rate $R_{\textrm{BS}}$;
that is, we show the case where the backhaul links become a
bottleneck. In the infrastructure-limited regime, either the ISH or
IMH protocol outperforms the other schemes while its throughput
scaling exponent depends on $\eta$. Two new operating regimes
$\tilde{\textrm{B}}$ and $\tilde{\textrm{D}}$ causing an
infrastructure limitation for some $\alpha$ are identified in
Table~\ref{Tab:RateEta}. More specifically, Regimes
$\tilde{\textrm{B}}$ and $\tilde{\textrm{D}}$ become
infrastructure-limited when $\alpha\geq 4-2\beta-2\eta$ and
$4-2\beta-2\eta \leq \alpha<2+\frac{2(\gamma-\eta)}{1-\beta}$,
respectively. These regimes are also DoF-limited since the
throughput scaling exponent is given by $\beta+\eta$ depending on
the number of BSs. In Regime $\tilde{\textrm{B}}$, the IMH protocol
is dominant when $\alpha\geq 4-2\beta-2\eta$.\footnote{For some case
in Regime $\tilde{\textrm{B}}$, the network using the ISH protocol
is also limited by the backhaul transmission, leading to the same
throughput scaling exponent $\beta+\eta$ as the IMH protocol case.}
In Regime $\tilde{\textrm{D}}$, the following interesting
observations are made according to the value of $\alpha$:
\begin{itemize}
\item (High path-loss attenuation regime) If $\alpha\geq
2+\frac{2(\gamma-\eta)}{1-\beta}$, then the network using the ISH
and IMH protocols is limited by the access and exit routings not by
the backhaul transmission, and thus achieves the same throughput as
that in Regime $\textrm{D}$.
\item (Medium path-loss attenuation regime) If $4-2\beta-2\eta \leq \alpha< 2+\frac{2(\gamma-\eta)}{1-\beta}$, then the
network using the ISH protocol is limited by the backhaul
transmission but achieves a higher throughput than those of pure ad
hoc routings, which is thus in the infrastructure-limited regime.
The network using the IMH protocol is not limited by the backhaul
transmission and thus its throughput scaling exponent is always less
than $\beta+\eta$.
\item (Low path-loss attenuation regime) If $\alpha<4-2\beta-2\eta$, neither the ISH nor IMH protocol can outperform the HC strategy since long-range MIMO transmissions of the HC yields a significant gain for small $\alpha$.
\end{itemize}

\begin{table}[t]
    \centering
    \caption{Achievability result for a hybrid extended network with finite-capacity infrastructure}
    \label{Tab:RateEta}
\begin{tabular}{|c|c|c|c|c|}
  \hline
  Regime & Range of $\eta$ & Condition & Best scheme & Throughput scaling \\ & & & & exponent $e$\\
  \noalign{\hrule height 1.2pt}
  \hline
  \multirow{2}{*}{$\tilde{\textrm{B}}$} & \multirow{2}{*}{$-\frac{1}{2}\leq\eta<\frac{1}{2}$}
                     & $2<\alpha<4-2\beta-2\eta$   & HC  & $2-\frac{\alpha}{2}$\\
                     && $\alpha\geq 4-2\beta-2\eta$ & IMH & $\beta+\eta$\\
  \hline
  \multirow{4}{*}{$\tilde{\textrm{D}}$} & \multirow{4}{*}{$0\leq\eta<1$} & $2<\alpha<4-2\beta-2\eta$  & HC & $2-\frac{\alpha}{2}$\\
                     && $ 4-2\beta-2\eta \leq \alpha<2+\frac{2(\gamma-\eta)}{1-\beta}$ & ISH & $\beta+\eta$\\
                     && $2+\frac{2(\gamma-\eta)}{1-\beta}\leq\alpha<1+\frac{2\gamma}{1-\beta}$ & ISH & $1+\gamma-\frac{\alpha(1-\beta)}{2}$\\
                     && $\alpha\geq 1+\frac{2\gamma}{1-\beta}$ & IMH & $\frac{1+\beta}{2}$\\
  \hline
\end{tabular}
\end{table}

If $\eta$ is too small, then some portions in Regimes B, C, and D
may turn into Regime A, where the pure ad hoc protocols outperform
the infrastructure-supported protocols, which will be specified
later. In these regimes, the throughput scaling is not improved even
with increasing $\eta$ and thus the network is not
infrastructure-limited. In Regimes B, C, and D, the
infrastructure-supported protocols can achieve their maximum
throughput, which is the same as the infinite-capacity backhaul link
case, since $\eta$ is sufficiently large. Hence, these three regimes
are not fundamentally infrastructure-limited.
\end{remark}

The two-dimensional operating regimes specified by $\beta$ and
$\gamma$ in Fig.~\ref{Fig:OperatingRegimesInfiniteBScapacity} can be
extended to three-dimensional operating regimes by introducing a new
scaling parameter $\eta$, where $R_{\textrm{BS}}=n^{\eta}$ for
$\eta\in(-\infty,\infty)$. Since the three-dimensional operating
regimes cannot be straightforwardly illustrated and even a
three-dimensional representation does not lead to any insight into
our analytically intractable network model, we identify the
three-dimensional operating regimes by introducing five types of
two-dimensional operating regimes, showing different
characteristics, with respect to $\beta$ and $\gamma$ according to
the value of $\eta$.

\begin{remark}[Three-dimensional operating regimes]
The operating regimes with respect to $\beta$ and $\gamma$ are
plotted in Figs.
\ref{Fig:OperatingRegimeEta-1}--\ref{Fig:OperatingRegimeEta-4} for
$\eta <-1/2$, $-1/2\leq\eta <0$, $0\leq\eta<1/2$, and
$1/2\leq\eta<1$, respectively. This result is analyzed in Appendix
\ref{Appendix-ThroughputOperatingRegimes}. In these figures, the
infrastructure-limited regimes are marked with a shaded area. Let us
closely scrutinize each case.

\begin{figure}[t!]
  \centering
  \leavevmode \epsfxsize=4.8in
  \epsffile{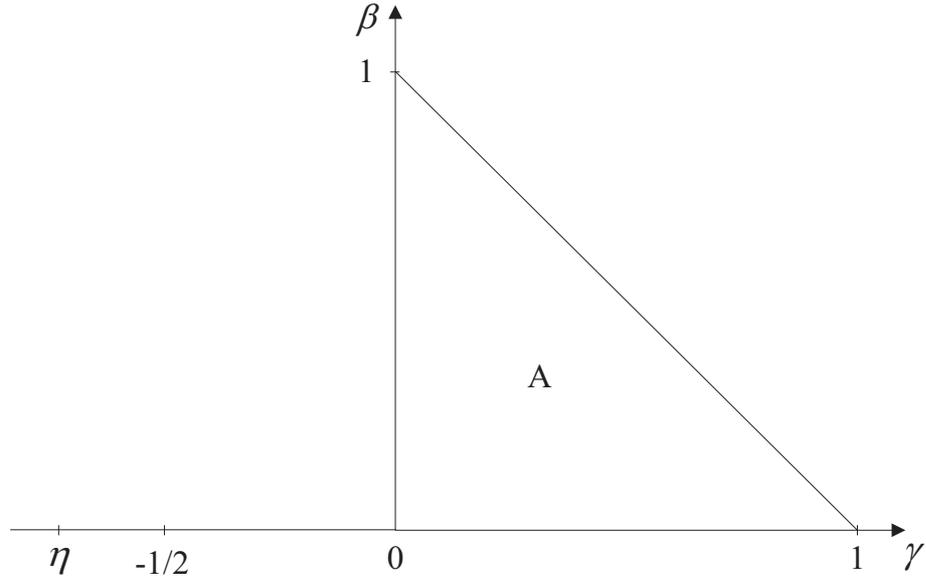}
  \caption{The operating regime with respect to $\beta$ and $\gamma$, where $\eta<-\frac{1}{2}$.}
  \label{Fig:OperatingRegimeEta-1}
\end{figure}

\begin{figure}[t!]
  \centering
  \leavevmode \epsfxsize=4.7in
  \epsffile{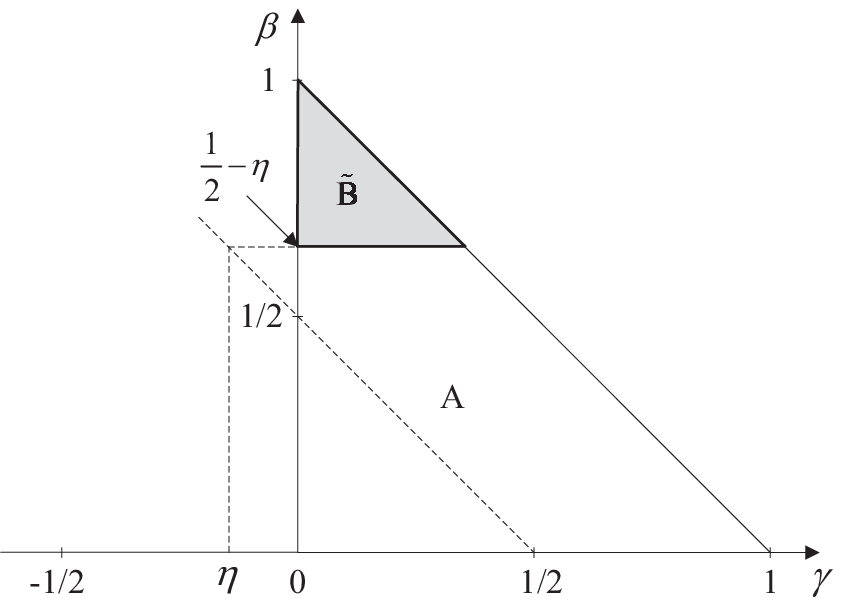}
  \caption{The operating regimes with respect to $\beta$ and $\gamma$, where $-\frac{1}{2}\leq \eta<0$.}
  \label{Fig:OperatingRegimeEta-2}
\end{figure}

\begin{figure}[t!]
  \centering
  \leavevmode \epsfxsize=3.8in
  \epsffile{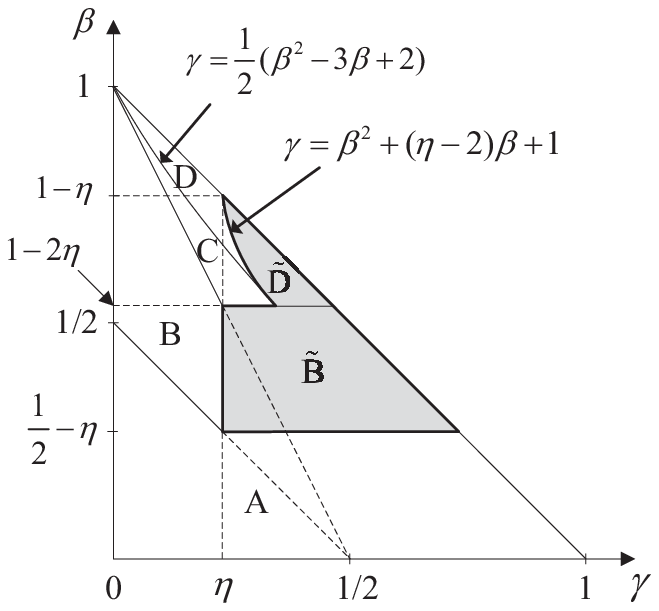}
  \caption{The operating regimes with respect to $\beta$ and $\gamma$, where $0\leq \eta <\frac{1}{2}$.}
  \label{Fig:OperatingRegimeEta-3}
\end{figure}

\begin{figure}[h]
  \centering
  \leavevmode \epsfxsize=3.6in
  \epsffile{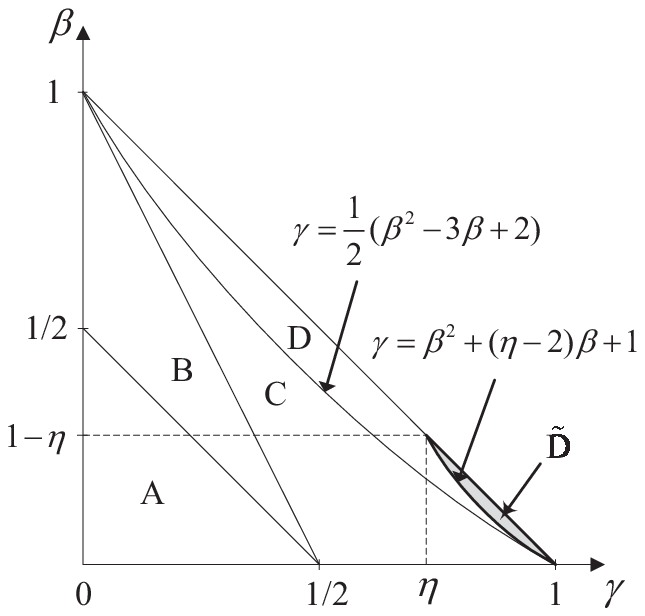}
  \caption{The operating regimes with respect to $\beta$ and $\gamma$, where $\frac{1}{2}\leq \eta <1$.}
  \label{Fig:OperatingRegimeEta-4}
\end{figure}

\begin{itemize}
\item $\eta<-\frac{1}{2}$: As shown in Fig. \ref{Fig:OperatingRegimeEta-1}, when $\eta$ is less
than $-1/2$, the entire regimes are included in Regime A. This
indicates that the infrastructure does not improve the capacity
scaling if the backhaul link rate scales slower than $1/\sqrt{n}$,
i.e., $R_{\textrm{BS}}=o(1/\sqrt{n})$.

\item $-\frac{1}{2}\leq \eta<0$: As $\eta$ becomes greater than $-1/2$, the infrastructure can
improve the capacity scaling for some cases but the network is
limited by the backhaul transmission, thereby resulting in Regime
$\tilde{\textrm{B}}$ (see Fig. \ref{Fig:OperatingRegimeEta-2}). In
Regime $\tilde{\textrm{B}}$, the HC protocol exhibits the best
performance when $\alpha$ is small. As $\alpha$ increases, the
throughput achieved by the HC protocol decreases due to the penalty
for long-range MIMO transmissions. The IMH protocol via backhaul
links then becomes dominant.

\item $0\leq \eta <\frac{1}{2}$: If $\eta$ is greater than zero, then the infrastructure-supported
protocols can fully achieve their throughput as in the network with
infinite-capacity infrastructure in Regimes B, C, and D. However, in
Regimes $\tilde{\textrm{B}}$ and $\tilde{\textrm{D}}$, the network
using either the ISH or IMH protocol is still limited by the
backhaul transmission (specifically when the IMH in Regime
$\tilde{\textrm{B}}$ or the ISH in Regime $\tilde{\textrm{D}}$ is
used). We refer to Fig.~\ref{Fig:OperatingRegimeEta-3}.

\item $\frac{1}{2}\leq \eta <1$: As
$\eta$ further increases beyond $1/2$, Regime $\tilde{\textrm{B}}$
disappears since the network using the IMH protocol is not limited
by backhaul transmission anymore and the area of Regime
$\tilde{\textrm{D}}$ gets reduced (see Fig.
\ref{Fig:OperatingRegimeEta-4}).

\item $\eta \ge 1$: As long as $\eta$ is greater than or equal to 1, i.e., $R_{\textrm{BS}}=\Omega(n)$, the network has no
infrastructure limitation at all, while achieving the same
throughput scaling as in the infinite-capacity backhaul link case.
The associated operating regimes are then illustrated in
Fig.~\ref{Fig:OperatingRegimesInfiniteBScapacity}.
\end{itemize}
\end{remark}

%\begin{figure}[t!]
%    \centering
%  % Requires \usepackage{graphicx}
%  \includegraphics[width=4.5in]{Figures/Inequalities_total_20130706}\\
%  \caption{The scaling exponents of the achievable schemes with respect to $\alpha$.}
%  \label{Fig:InequalityEta}
%\end{figure}

%%%%%%%%%%%%%%%%%%%%%%%%%%%%%%%%%%%%%%%%%%%%%%%%%%%%%%%%%%%%%%%%%%%%%%%%%%%%%%%%%%%%%%%%%%%%%%%%%%%%%%%%%%%%%%%%%%%%%%%%%%%%%%%%%%%%%%%%%%%%%%%%%%%%%

\section{Cut-Set Upper Bound}\label{SEC:CutSetUpperBound}

\begin{figure}[t!]
  \centering
  {
  \subfigure[The cut $L_1$]{\leavevmode
  \leavevmode \epsfxsize=0.465\textwidth
  \epsffile{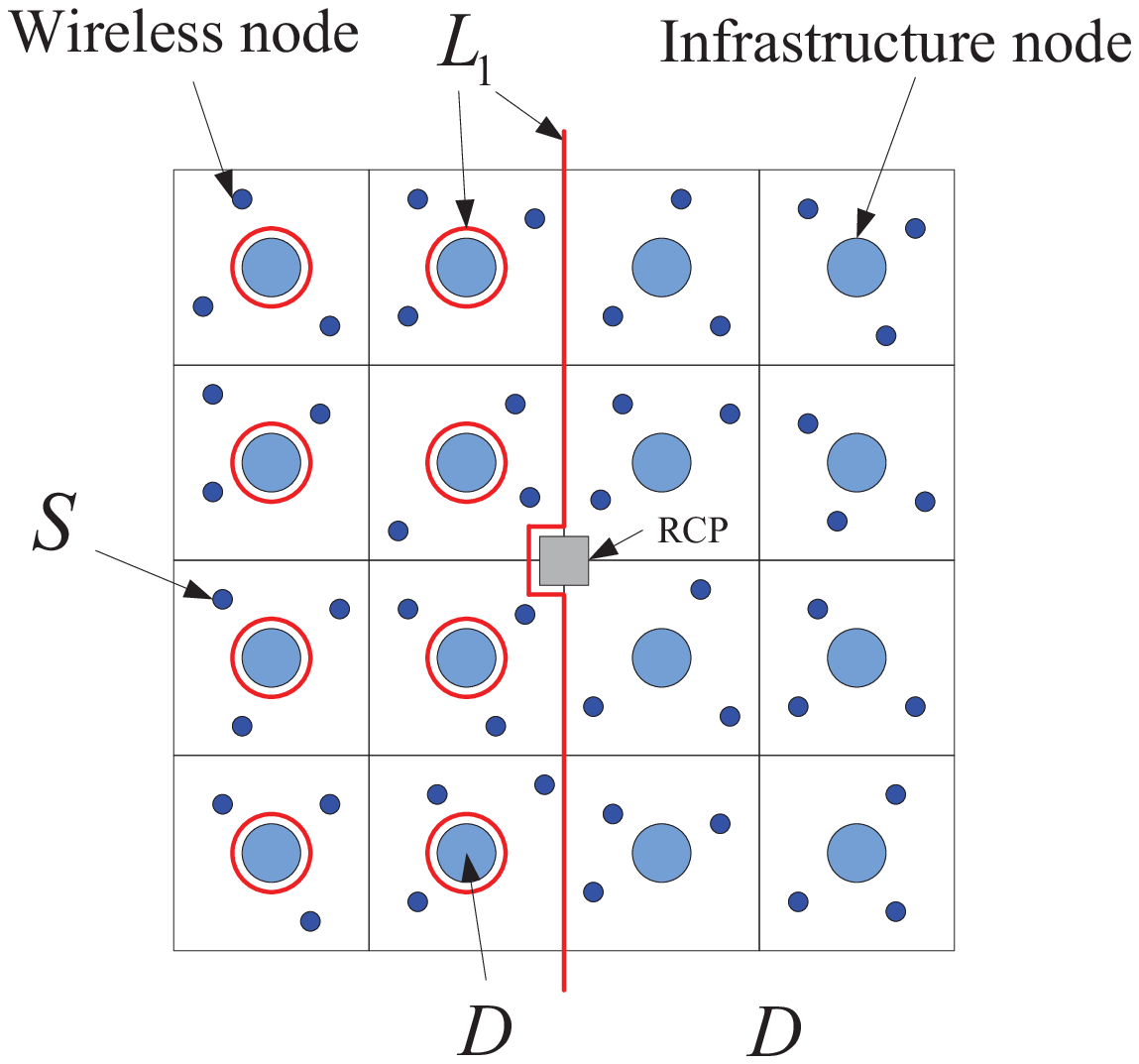} \label{FIG:Cut1}
  }
  \hspace{0.2in}
  \subfigure[The cut $L_2$]{\leavevmode
  \leavevmode \epsfxsize=0.358\textwidth
  \epsffile{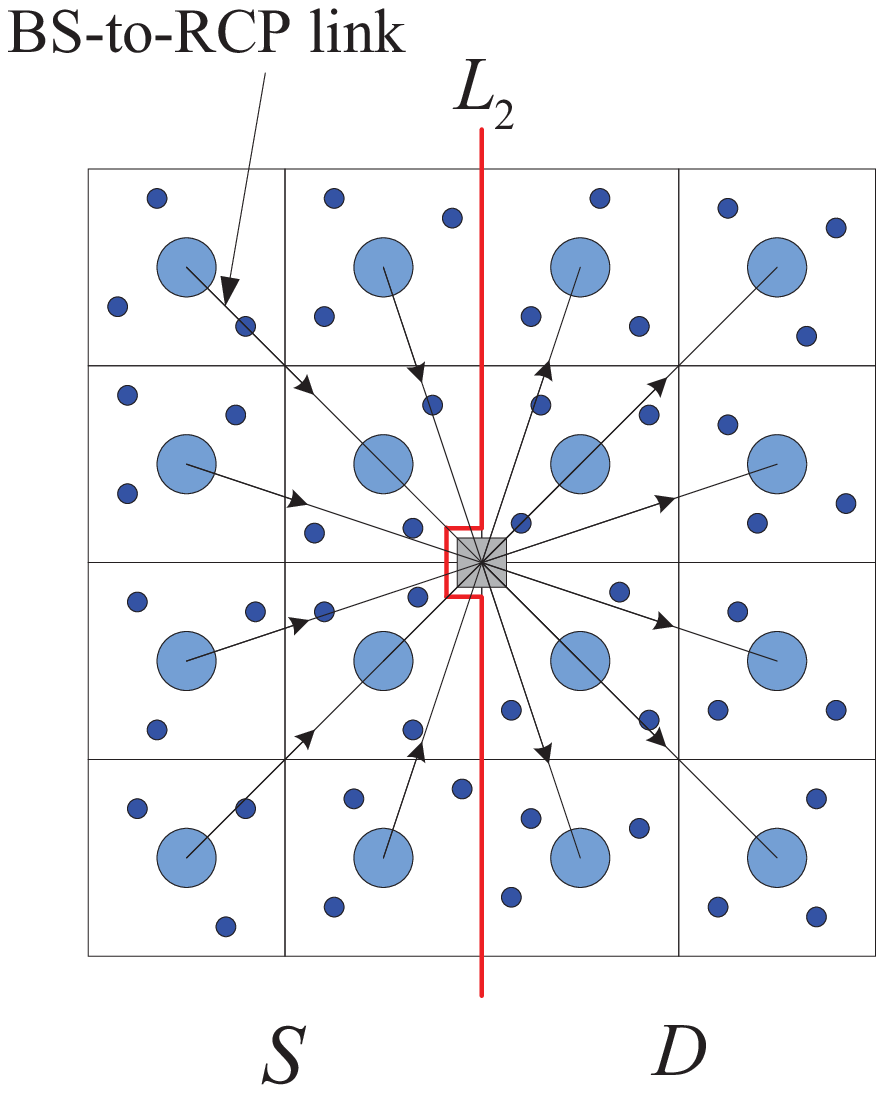} \label{FIG:Cut2}
  }
  }
  \caption{The cuts $L_1$ and $L_2$ in the hybrid extended network. The BS-to-RCP or RCP-to-BS links are not shown in (a) since they are not in effect under $L_1$.}
  \label{FIG:Cut}
\end{figure}

In this section, to see how closely our achievable scheme approaches
the fundamental limit in the hybrid network with rate-limited
BS-to-RCP (or RCP-to-BS) links, a generalized cut-set upper bound on
the aggregate capacity scaling based on the information-theoretic
approach is derived. As illustrated in Fig. \ref{FIG:Cut}, in order
to provide a tight upper bound, two cuts $L_1$ and $L_2$ are taken
into account. Similarly as
in~\cite{ShinJeonDevroyeVuChungLeeTarokh:08}, the cut $L_1$ divides
the network area into two halves by cutting the wireless connections
between wireless source nodes on the left of the network and the
other nodes, including all BS antennas and one RCP. In addition, to
fully utilize the main characteristics of the network with
finite-capacity infrastructure, we consider another cut $L_2$, which
divides the network area into another two halves by cutting the
wired connections between BSs and the RCP as well as the wireless
connections between all nodes (including BS antennas) located on the
left of the network and all nodes (including BS antennas and the
RCP) on the right.

Upper bounds obtained under the cuts $L_1$ and $L_2$ are denoted by
$T_n^{(1)}$ and $T_n^{(2)}$, respectively. By the cut-set theorem,
the total capacity is upper-bounded by
\begin{align}
T_n\le \min\left\{T_n^{(1)}, T_n^{(2)}\right\}. \label{EQ:mincut}
\end{align}
%The cut $L_1$ was originally used to characterize the upper bound on
%the capacity scaling of ad hoc networks with infinite-capacity
%infrastructure. The cut $L_2$ is first used to take a finite
%capacity scaling of infrastructures into consideration.
The following two lemmas will be used to derive an upper bound on
the capacity in the remainder of this section.

\begin{figure}[t!]
  \centering
  {
  \subfigure[The cut $L_1$]{\leavevmode
  \leavevmode \epsfxsize=0.35\textwidth
  \epsffile{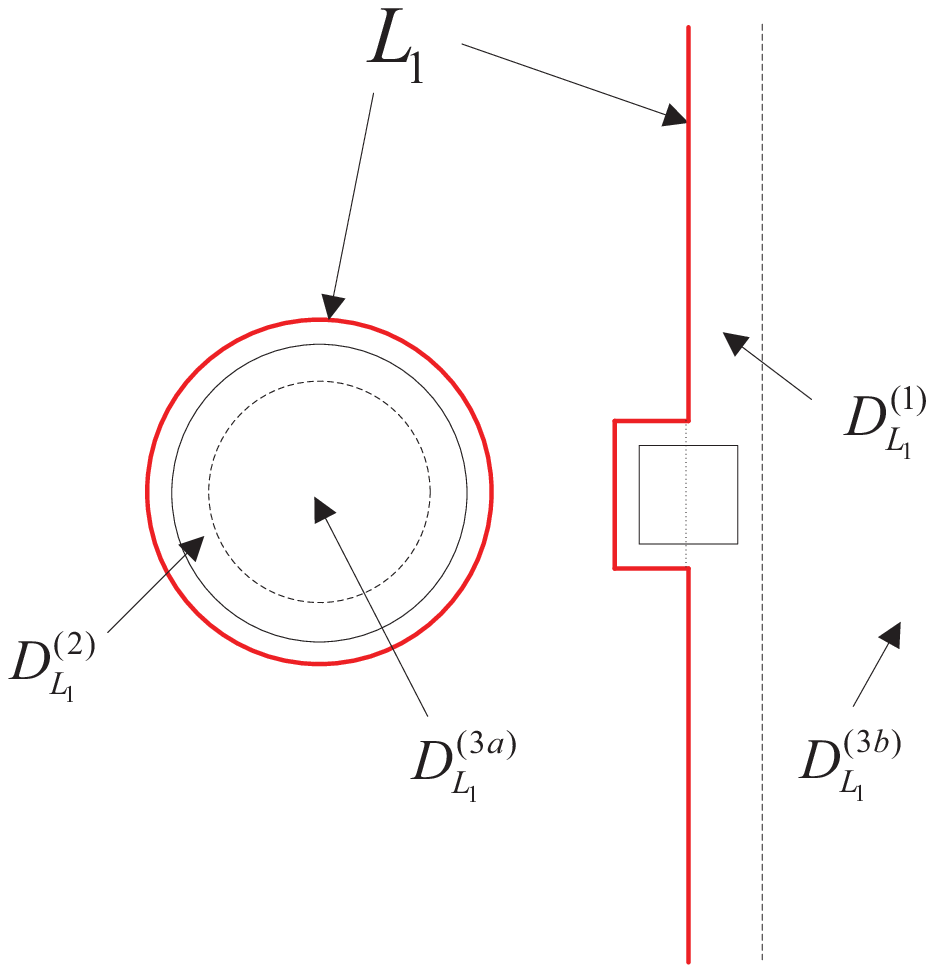} \label{FIG:CutDestination1}
  }
  \hspace{0.75in}
  \subfigure[The cut $L_2$]{\leavevmode
  \leavevmode \epsfxsize=0.38\textwidth
  \epsffile{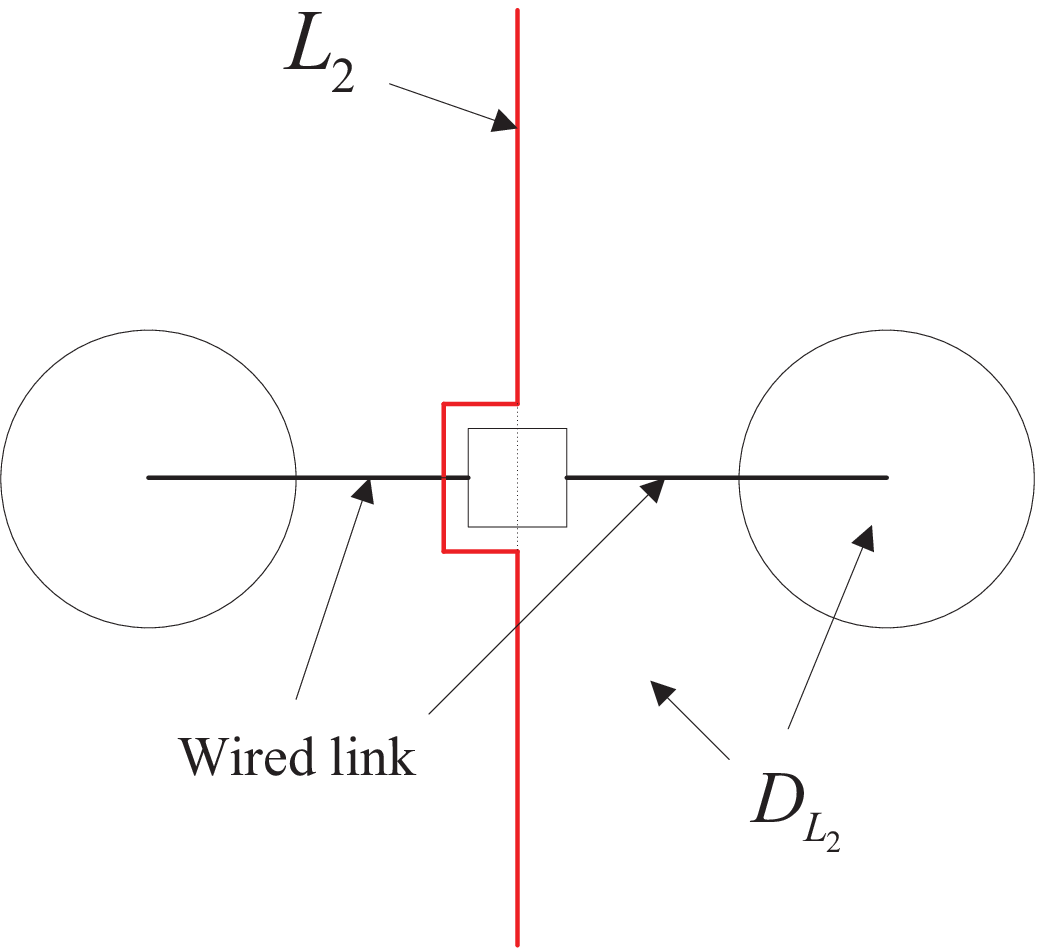} \label{FIG:CutDestination2}
  }
  }
  \caption{The partition of destinations under the cuts $L_1$ and $L_2$ in the hybrid extended network. To simplify the figure, one and two BSs are shown in (a) and (b) with the RCP, respectively.}
  \label{FIG:CutDestination}
\end{figure}

\begin{lemma}\label{Lem:minDistance}
In our two-dimensional extended network, where $n$ nodes are
uniformly distributed and there are $m$ BSs with $l$ regularly
spaced antennas, the minimum distance between any two nodes or
between a node and an antenna on the BS boundary is greater than
$1/n^{1/2+\epsilon_1}$ whp for an arbitrarily small $\epsilon_1>0$.
\end{lemma}

This lemma can be obtained by the derivation similar to that of
Lemma 6 in~\cite{ShinJeonDevroyeVuChungLeeTarokh:08}.

\begin{lemma}\label{Lem:MaxNumNodesInSquare}
Assume a two-dimensional extended network. When the network area
with the exclusion of BS area is divided into $n$ squares of unit
area, there are less than $\log n$ nodes in each square whp.
\end{lemma}

The proof of Lemma~\ref{Lem:MaxNumNodesInSquare} is given by
slightly modifying the proof of in~\cite[Lemma
1]{FranceschettiDouseTseThiran:07}. Let us first focus on the cut
$L_1$. Let $S_{L_1}$ and $D_{L_1}$ denote the sets of sources and
destinations, respectively, for $L_1$ in the network. Then, all
wireless ad hoc nodes on the left half of the network are $S_{L_1}$,
while all ad hoc nodes on the right half and all BS antennas in the
network are destinations $D_{L_1}$ (see Fig.~\ref{FIG:Cut1}). Note
that the wired BS-to-RCP (or RCP-to-BS) links do not need to be
considered under $L_1$ since all BSs in the network act as
destinations $D_{L_1}$. In this case, the $\frac{n}{2}\times
\left(\frac{n}{2}+ml\right)$ MIMO channel between the two sets of
nodes and BSs separated by the cut $L_1$ is formed. The total
throughput $T_n^{(1)}$ for sources on the left half is bounded by
the capacity of the MIMO channel between $S_{L_1}$ and $D_{L_1}$,
and thus is given by
\begin{align*}
    T_n^{(1)}&\leq \underset{\mathbf{Q}_{L_1}\geq 0}\max~E\left[\log\det\left(\mathbf{I}_{\frac{n}{2}+ml}+\mathbf{H}_{L_1}\mathbf{Q}_{L_1}\mathbf{H}_{L_1}^\dag\right)\right]
    \nonumber\\
    &=\underset{\mathbf{Q}_{L_1}\geq 0}\max~ E[\log\det(\mathbf{I}_{\Theta(n)}+\mathbf{H}_{L_1}\mathbf{Q}_{L_1}\mathbf{H}_{L_1}^\dag)]
\end{align*}
where the equality comes from the fact that
$n=\Omega(ml)$.\footnote{Here and in the sequel, the noise variance
is assumed to be one to simplify the notation.} Here, the channel
matrix $\mathbf{H}_{L_1}$ consists of the uplink channel vectors
$\mathbf{h}_{bi}^{(u)}$  in (\ref{EQ:uplinkCH}) for $i\in S_{L_1}$,
$b\in B$, and $h_{ki}$ in (\ref{EQ:nodeCH}) for $i\in S_{L_1}$,
$k\in D_r$, where $B$ and $D_r$ denotes the set of BSs in the
network and the set of wireless nodes on the right half,
respectively. The matrix $\mathbf{Q}_{L_1}$ is the positive
semidefinite input covariance matrix whose $k$th diagonal element
satisfies $[\mathbf{Q}_{L_1}]_{kk}\leq P$ for $k\in S_{L_1}$. In
order to obtain a tight upper bound, it is necessary to narrow down
the class of S--D pairs according to their Euclidean distances. As
illustrated in Fig.~\ref{FIG:CutDestination1}, the set $D_{L_1}$ is
partitioned into the following four groups according to their
locations: $D_{L_1}^{(1)}$, $D_{L_1}^{(2)}$, $D_{L_1}^{(3a)}$, and
$D_{L_1}^{(3b)}$. The set $D_{L_1}^{(3a)}\cup D_{L_1}^{(3b)}$ is
denoted by $D_{L_1}^{(3)}$. The sets $D_{L_1}^{(1)}$ and
$D_{L_1}^{(2)}$ represent the sets of destinations located on the
rectangular slab with width one immediately to the right of the
centerline (cut) $L_1$ including the RCP and on the ring with width
one immediately inside each BS boundary (cut) on the left half,
respectively. The set $D_{L_1}^{(3)}$ is given by
$D_{L_1}\setminus(D_{L_1}^{(1)} \cup D_{L_1}^{(2)})$. By generalized
Hadamard's inequality~\cite{ConstantinescuScharf:98} as
in~\cite{OzgurLevequeTse:07,JovicicViswanathKulkarni:04}, we then
have
\begin{align}\label{Eq:ThreeTermsUpperBoundL1}
    T_n^{(1)}&\leq\underset{\mathbf{Q}_{L_1}\geq 0}\max~ E\left[\log\det\left(\mathbf{I}_{\sqrt{n}\log n+1}+\mathbf{H}_{L_1}^{(1)}\mathbf{Q}_{L_1}\mathbf{H}_{L_1}^{(1)\dag}\right)\right]
    \nonumber\\
    &~~~+\underset{\mathbf{Q}_{L_1}\geq 0}\max~ E\left[\log\det\left(\mathbf{I}_{O(\sqrt{mn})}+\mathbf{H}_{L_1}^{(2)}\mathbf{Q}_{L_1}\mathbf{H}_{L_1}^{(2)\dag}\right)\right]
    \nonumber\\
    &~~~+\underset{\mathbf{Q}_{L_1}\geq 0}\max~ E\left[\log\det\left(\mathbf{I}_{\Theta(n)}+\mathbf{H}_{L_1}^{(3)}\mathbf{Q}_{L_1}\mathbf{H}_{L_1}^{(3)\dag}\right)\right].
\end{align}
where $\mathbf{H}_{L_1}^{(t)}$ is the matrix with entries
$\left[\mathbf{H}_{L_1}^{(t)}\right]_{ki}$ for $i\in S_{L_1}$, $k\in
D_{L_1}^{(t)}$, and $t=1,2,3$. By analyzing either the sum of the
capacities of the multiple-input single-output (MISO) channel or the
amount of power transferred across the network, an upper bound for
each term in (\ref{Eq:ThreeTermsUpperBoundL1}) can be derived. We
now establish the following lemma, which shows an upper bound under
the cut $L_1$.

\begin{lemma}\label{Lem:UpperBoundL1}
Under the cut $L_1$ in Fig. \ref{FIG:CutDestination1}, an upper
bound on the aggregate capacity, $T_n^{(1)}$, of the hybrid network
with rate-limited infrastructure is given by
\begin{align}\label{Eq:UpperBoundL1}
    &T_n^{(1)}= O\left(n^{\epsilon}\max \left\{ml\left(\frac{m}{n}\right)^{\alpha/2-1},m\min\left\{l,\sqrt{\frac{n}{m}}\right\},
    \right.\right.
    \nonumber\\
    &\left.\left.
    ~~~~~~~~~~~~\sqrt{n},n^{2-\alpha/2}\right\}\right),
\end{align}
where $\epsilon>0$ is an arbitrarily small constant.
%\footnote{The
%first to fourth terms in (\ref{Eq:UpperBoundL1}) can be achieved by
%the ISH, IMH, MH, and HC protocols, respectively.}
\end{lemma}

\begin{proof}
Let $T_{n,1}^{(1)}$, $T_{n,2}^{(1)}$, and $T_{n,3}^{(1)}$ denote the
first to third terms in (\ref{Eq:ThreeTermsUpperBoundL1}),
respectively. By further applying generalized Hadamard's
inequality~\cite{ConstantinescuScharf:98}, the term $T_{n,1}^{(1)}$
is upper-bounded by
\begin{align}
    T_{n,1}^{(1)}&\leq \sum_{k\in D_{L_1}^{(1)}}\log \left(1+P\sum_{i\in S_{L_1}}|h_{ki}|^2\right)\nonumber\\
    &\leq c_1 (\sqrt{n}\log n+1)\log n
    \leq c_2 n^{\epsilon}\sqrt{n} \label{EQ:Tn1}
\end{align}
where $c_1$ and $c_2$ are some positive constants, independent of
$n$, and $\epsilon>0$ is an arbitrarily small constant. The second
inequality follows from the fact that the minimum distance between
any source and destination (including the RCP) is greater than
$1/n^{1/2+\epsilon_1}$ whp for an arbitrarily small $\epsilon_1>0$
by Lemma~\ref{Lem:minDistance} and there exist no more than
$\sqrt{n}\log n +1 $ nodes in $D_{L_1}^{(1)}$ whp by
Lemma~\ref{Lem:MaxNumNodesInSquare}.
%Note
%that $T_{n,1}^{(1)}$ is upper-bounded by the same $c_2
%n^{\epsilon}\sqrt{n}$ as
%in~\cite{ShinJeonDevroyeVuChungLeeTarokh:08} where there is no RCP
%in the network.
Since the remaining terms (i.e., the second and third terms) in
(\ref{Eq:ThreeTermsUpperBoundL1}) can be computed regardless of the
presence of the RCP, they are derived by basically following the
same approach as that in~\cite{ShinJeonDevroyeVuChungLeeTarokh:08}.
From the antenna configuration in our network, the second term in
(\ref{Eq:ThreeTermsUpperBoundL1}), $T_{n,2}^{(1)}$, is bounded
by~\cite{ShinJeonDevroyeVuChungLeeTarokh:08}
\begin{align}
    T_{n,2}^{(1)}\leq c_3 m\min\left\{l,\sqrt{\frac{n}{m}}\right\}\log n
    \label{EQ:Tn2}
\end{align}
where $c_3>0$ is some constant independent of $n$. An upper bound
for the third term in (\ref{Eq:ThreeTermsUpperBoundL1}),
$T_{n,3}^{(1)}$, is derived now. If $l=o(\sqrt{n/m})$, there is no
destination in $D_{L_1}^{(3a)}$ (see Fig.~\ref{FIG:CutDestination1})
and thus the information transfer to the set $D_{L_1}^{(3)}$ is the
same as that in wireless network with no infrastructure (that is,
the information transfer to the set $D_{L_1}^{(3b)}$). Hence, it
follows that~\cite{OzgurLevequeTse:07}
\begin{align}
    T_{n,3}^{(1)} = O\left(\max\left\{n^{2-\alpha/2+\epsilon},n^{1/2+\epsilon}\right\}\right)\label{EQ:Tn31}
\end{align}
If $l=\Omega(\sqrt{n/m})$, using an upper bound for the power
transfer from the set $S_{L_1}$ to the set $D_{L_1}^{(3a)}$, the
term $T_{n,3}^{(1)}$ is upper-bounded
by~\cite{ShinJeonDevroyeVuChungLeeTarokh:08}
\begin{align}
    T_{n,3}^{(1)}&\leq \left\{
    \begin{array}{ll}
        c_4 n^{\epsilon}\max\left\{n^{2-\alpha/2},nl\left(\frac{m}{n}\right)^{\alpha/2}\right\} & \textrm{if~}2<\alpha<3\\
        c_4 n^{\epsilon}\max\left\{\sqrt{n},\frac{n}{\sqrt{l}}\left(\frac{ml}{n}\right)^{\alpha/2}\right\} & \textrm{if~}\alpha\geq 3
    \end{array}
    \right. \label{EQ:Tn32}
\end{align}
where $c_4>0$ is some constant independent of $n$. Using
(\ref{EQ:Tn1})--(\ref{EQ:Tn32}) finally yields
(\ref{Eq:UpperBoundL1}), which completes the proof of the lemma.
\end{proof}

The upper bound $T_n^{(1)}$ matches the achievable throughput
scaling within a factor of $n^\epsilon$ in the network with
infinite-capacity infrastructure ($\eta\rightarrow\infty$), which
indicates that $T_n^{(1)}$ does not rely on the parameter $\eta$.
Note that the first to fourth terms in the max operation of
(\ref{Eq:UpperBoundL1}) represent the amount of information
transferred to the destination sets $D_{L_1}^{(3a)}$,
$D_{L_1}^{(2)}$, $D_{L_1}^{(1)}$, and $D_{L_1}^{(3b)}$,
respectively. The third and fourth terms characterize the cut-set
upper bound of wireless networks with no infrastructure. By
employing infrastructure nodes, it is possible to get an additional
information transfer for a given cut $L_1$, corresponding to the
first and second terms in the max operation of
(\ref{Eq:UpperBoundL1}).

In addition to the cut $L_1$, we now turn to the cut $L_2$ in
Fig.~\ref{FIG:Cut2} to obtain a tight cut-set upper bound in the
network with rate-limited infrastructure. Let $S_{L_2}$ and
$D_{L_2}$ denote the sets of sources and destinations, respectively,
for $L_2$ in the network. More precisely, under $L_2$, all the
wireless and infrastructure nodes on the left half are $S_{L_2}$,
while all the nodes on the right half including the RCP are included
in $D_{L_2}$ (see Fig.~\ref{FIG:Cut2}). Unlike the case of $L_1$, we
take into account information flows over the wired connections as
well as the wireless connections. In consequence, an upper bound on
the aggregate capacity is established based on using the min-cut of
our hybrid network, and is presented in the following theorem.

\begin{theorem}\label{Thm:TotalThroughputUpperBound}
In the hybrid network with the backhaul link rate $R_{\textrm{BS}}$,
the aggregate throughput $T_n$ is upper-bounded by
\begin{align}\label{Eq:UpperBound}
    &O\left(
    \max\left\{\min\left\{
    \max\left\{n^{\epsilon}ml\left(\frac{m}{n}\right)^{\alpha/2-1},
    \right.\right.\right.\right.
    \nonumber\\
    &\left.\left.
    ~~~~~~~~~~~~~~~~~~~~~~~~~~
    n^{\epsilon}m\min\left\{ l,\sqrt{\frac{n}{m}}\right\}\right\}
    ,m R_{\textrm{BS}}\right\},
    \nonumber\\
    &\left.\left.
    ~~~~~~~~~~~n^{1/2+\epsilon},n^{2-\alpha/2+\epsilon}\right\}
    \right),
\end{align}
where $\epsilon>0$ is an arbitrarily small constant.\footnote{To
simplify notations, the terms including $\epsilon$ are omitted if
dropping them does not cause any confusion.}
\end{theorem}

\begin{proof}
By the min-cut of the network, the aggregate capacity is
upper-bounded by (\ref{EQ:mincut}). The upper bound on the capacity,
$T_n^{(1)}$, under the cut $L_1$ directly follows from Lemma
\ref{Lem:UpperBoundL1}.

An upper bound on the capacity, $T_n^{(2)}$, under the cut $L_2$ is
derived below. Since the cut-set argument under $L_1$ does not
utilize the main characteristics of rate-limited backhaul links, we
now deal with a new cut $L_2$, which divides the network into two
equal halves. A wired link between two BSs that lie on the opposite
side of each other from the cut is illustrated in Fig.
\ref{FIG:CutDestination2}. In this case, we get the $\left(\frac{n +
ml}{2}\right) \times \left(\frac{n + ml}{2}+1\right)$ MIMO {\it
wireless} channel and the $\frac{m}{2} \times 1$ MISO {\it wired}
channel between the two sets $S_{L_2}$ and $D_{L_2}$ separated by
$L_2$. We now derive the amount of information transferred by each
channel. Let $T^{(2)}_{n,\textrm{wireless}}$ and
$T^{(2)}_{n,\textrm{wired}}$ denote the amount of information
transferred through the wireless and wired channels under $L_2$,
respectively. The aggregate capacity obtained under $L_2$ is then
upper-bounded by $T^{(2)}_n \leq
T^{(2)}_{n,\textrm{wireless}}+T^{(2)}_{n,\textrm{wired}}$.

Let us first focus on deriving $T^{(2)}_{n,\textrm{wireless}}$,
which is bounded by the capacity of the MIMO channel between
$S_{L_2}$ and $D_{L_2}$. It thus follows that
\begin{align*}
    T^{(2)}_{n,\textrm{wireless}} &\leq \underset{\mathbf{Q}_{L_2}\geq 0}\max~ E\left[\log\det\left(\mathbf{I}_{\frac{n+ml}{2}+1}+\mathbf{H}_{L_2}\mathbf{Q}_{L_2}\mathbf{H}_{L_2}^\dag\right)\right]
    \nonumber\\
    &=\underset{\mathbf{Q}_{L_2}\geq 0}\max~
    E[\log\det(\mathbf{I}_{\Theta(n)}+\mathbf{H}_{L_2}\mathbf{Q}_{L_2}\mathbf{H}_{L_2}^\dag)],
\end{align*}
where the equality comes from $n=\Omega(ml)$. Using the same power
transfer argument as~\cite{OzgurLevequeTse:07}, an upper bound on
$T^{(2)}_{n,\textrm{wireless}}$ is derived in the following lemma.

\begin{lemma}
Under the cut $L_2$ in Fig. \ref{FIG:Cut2}, an upper bound
$T^{(2)}_{n,\textrm{wireless}}$ on the capacity of the
$\left(\frac{n + ml}{2}\right) \times \left(\frac{n +
ml}{2}+1\right)$ MIMO wireless channel is given by
\begin{align}
    T^{(2)}_{n,\textrm{wireless}} = O\left(n^\epsilon \max\left\{ \sqrt{n}, n^{2-\alpha/2}
    \right\}\right), \nonumber
\end{align}
where $\epsilon>0$ is an arbitrarily small constant.
\end{lemma}

This result can be obtained by straightforwardly applying the proof
technique in \cite{OzgurLevequeTse:07}. This is because the above
problem turns into showing the capacity of the classical $n \times
n$ MIMO wireless channel formed in pure ad hoc networks due to the
fact that $\frac{n + ml}{2}=\Theta(n)$ and the minimum distance
between a node and an antenna on the BS boundary is the same as the
minimum distance between any two nodes, which thus does not
essentially change the result in terms of scaling law.

Next, we turn to analyzing the capacity of the MISO wired channel.
There also exist wired links between the two sets separated by
$L_2$, i.e., BSs on the left half of the network and the RCP placed
at the center of the network. From the fact that the $\frac{m}{2}$
BS-to-RCP (or RCP-to-BS) links can be created under $L_2$, we have
\begin{align}
T^{(2)}_{n,\textrm{wired}} = O(m R_{\textrm{BS}} ) \nonumber
\end{align}
since the capacity of each backhaul link is assumed to be
$R_{\textrm{BS}}$. In consequence, under the cut $L_2$, it follows
that
\begin{align} \label{Eq:UpperBoundL2}
    T^{(2)}_n &\le T^{(2)}_{n,\textrm{wireless}} +
    T^{(2)}_{n,\textrm{wired}}
    \nonumber\\&
    = O\left(\max\left\{ m R_{\textrm{BS}}, n^{1/2+\epsilon},n^{2-\alpha/2+\epsilon}
    \right\}\right).
\end{align}

Using (\ref{Eq:UpperBoundL1}) and (\ref{Eq:UpperBoundL2}), the total
throughput $T_n$ is finally upper-bounded by
\begin{align}
    T_n
    \le \min\{T_n^{(1)},T_n^{(2)}\}
    \nonumber
\end{align}
\begin{align}
    &= O\left(\min\left\{n^{\epsilon}\max \left\{ml\left(\frac{m}{n}\right)^{\alpha/2-1},m\min\left\{l,\sqrt{\frac{n}{m}}\right\},
    %\right.\right.\right.
    %\nonumber\\
    %&~~\left.\left.\left.
    \right.\right.\right.
    \nonumber\\
    &\left.\left.\left.
    ~~~~~~~~~~~~\sqrt{n},n^{2-\alpha/2}\right\},
    \max\{m R_{\textrm{BS}},n^{1/2+\epsilon},n^{2-\alpha/2+\epsilon}\}\right\}\right)
    \nonumber\\
    &=O\left(\max\left\{\min\left\{\max \left\{n^{\epsilon}ml\left(\frac{m}{n}\right)^{\alpha/2-1},
    \right.\right.\right.\right.
    \nonumber\\
    &\left.\left.\left.\left.
    ~~~~~~n^{\epsilon}m\min\left\{l,\sqrt{\frac{n}{m}}\right\}
    \right\},m R_{\textrm{BS}}\right\},
    n^{1/2+\epsilon},n^{2-\alpha/2+\epsilon}\right\}\right),
    \nonumber
\end{align}
where the last equality follows from
$\min\{\max\{a,x\},\max\{b,x\}\}=\max\{\min\{a,b\},x\}$. This
completes the proof of Theorem~\ref{Thm:TotalThroughputUpperBound}.
\end{proof}

The relationship between the achievable throughput and the cut-set
upper bound is now examined as follows.

\begin{remark}
The upper bound in (\ref{Eq:UpperBound}) matches the achievable
throughput scaling in Theorem \ref{Thm:AchievableRateLimited} within
$n^{\epsilon}$ for an arbitrarily small $\epsilon>0$ in the hybrid
extended network with the finite backhaul link rate
$R_{\textrm{BS}}$. In other words, choosing the best of the four
achievable schemes ISH, IMH, MH, and HC is order-optimal for all the
operating regimes (even if the rate of each backhaul link between a
BS and the RCP is finite).
%Hence, any other complex routing protocol is not needed even in the
%network where the backhaul link rate is limited.
Note that the operating regimes on the upper bound in (\ref
{Eq:UpperBound}) are basically the same as those on the achievable
throughput illustrated in Figs.
\ref{Fig:OperatingRegimeEta-1}--\ref{Fig:OperatingRegimeEta-4}. Let
us now examine how to achieve each term in (\ref{Eq:UpperBound}).
The first and second terms in the max operation of
(\ref{Eq:UpperBound}) correspond to the rate scaling achieved the
ISH and IMH protocols for the infinite-capacity backhaul link case
(i.e., $\eta\rightarrow\infty$), respectively. The third term
represents the total transmission rate through the backhaul links
between all BSs and one RCP when the rate of each link is limited by
$R_{\textrm{BS}}$. The last two terms in the max operation of
(\ref{Eq:UpperBound}) correspond to the rate scaling achieved by the
MH and HC schemes without infrastructure support.
\end{remark}

Now we would like to examine in detail the amount of information
transfer by each separated destination set.

\begin{remark}
As mentioned earlier, each term in (\ref{Eq:UpperBound}) is
associated with one of the destination sets for a given cut. To be
concrete, the first and the second terms in the max operation of
(\ref{Eq:UpperBound}) represent the amount of information
transferred to $D_{L_1}^{(3a)}$ and $D_{L_1}^{(2)}$ over the {\em
wireless} connections, respectively. The third term represents the
information flows over the {\em wired} connections from the BSs on
the left half to the BSs on the right half. The fourth and fifth
terms represent the amount of information transferred to
$D_{L_1}^{(1)}$ and $D_{L_1}^{(3b)}$ over the {\em wireless}
connections, respectively.
\end{remark}

In Remarks \ref{Rem:DoF-limitedRegimes} and
\ref{Rem:Infra-limitedRegimes}, we showed the case where our network
is either in the DoF- or infrastructure-limited regime according to
the achievable throughput scaling. In a similar fashion, the
identical regimes that lead to such fundamental limitations can also
be identified using the cut-set argument drawn under $L_1$ and $L_2$
in the following remarks.

\begin{remark}[DoF-limited regimes]
Similarly as in Remark~\ref{Rem:DoF-limitedRegimes} obtained based
on the achievability result, the DoF-limited regime can also be
scrutinized by using the cut-set bound result in
Theorem~\ref{Thm:TotalThroughputUpperBound}.
\end{remark}

\begin{remark}[Infrastructure-limited regimes]
In Remark~\ref{Rem:Infra-limitedRegimes} obtained based on the
achievability result, the infrastructure-limited regime can also be
scrutinized using the upper bound derived under the two cuts.
\end{remark}

\section{Conclusion}\label{SEC:Conclusion}
A generalized capacity scaling was characterized for an
infrastructure-supported extended network assuming an arbitrary rate
scaling of backhaul links. The minimum required rate of each
backhaul link was first derived to guarantee the optimal capacity
scaling along with a cost-effective backhaul solution. Provided
three scaling parameters (i.e., the number of BSs, $m$, and the
number of antennas at each BS, $l$, the rate of each backhaul link,
$R_{\textrm{BS}}$) scale at arbitrary rates relative to the number
of wireless nodes, $n$, a generalized achievable throughput scaling
was then derived based on using one of the two
infrastructure-supported routing protocols, ISH and IMH, and the two
ad hoc routing protocols, MH and HC. Three-dimensional operating
regimes were also explicitly identified according to the three
scaling parameters. In particular, we studied the case where our
network is fundamentally DoF- or infrastructure-limited, or
doubly-limited. Furthermore, to show the optimality of our
achievability result, a generalized information-theoretic cut-set
upper bound was derived using two cuts. In the hybrid network with
rate-limited infrastructure, it was shown that the upper bound
matches the achievable throughput scaling for all the
three-dimensional operating regimes.

\appendix
\section{Appendix}
\renewcommand\theequation{\Alph{section}.\arabic{equation}}
\setcounter{equation}{0} \subsection{Three-Dimensional Operating
Regimes on the Achievable Throughput Scaling}
%\section{The Proof of Lemma~\ref{}}
\label{Appendix-ThroughputOperatingRegimes} Let $e_{\textrm{ISH}}$,
$e_{\textrm{IMH}}$, $e_{\textrm{MH}}$, and $e_{\textrm{HC}}$ denote
the scaling exponents for the total throughput achieved by using the
ISH, IMH, MH, and HC protocols, respectively. The scaling exponent
for the throughput achieved by the ISH and IMH protocols is given by
$\beta+\eta$ when the performance is limited by backhaul
transmission rate. In the following, the operating regimes with
respect to $\beta$ and $\gamma$ will be analyzed for each range of
the value of $\eta$ (the terms including $\epsilon$ are omitted for
notational convenience). For each case, we will check whether
$\max\{e_{\textrm{ISH}},e_{\textrm{IMH}}\}$ is given by
$\beta+\eta$. If so, it will be compared with
$\max\{e_{\textrm{MH}},e_{\textrm{HC}}\}$ to determine whether the
regime is infrastructure-limited or Regime A (i.e., no
infrastructure is needed). Otherwise, the regimes will be the same
as those for the network with infinite-capacity infrastructure.

\subsubsection{$\eta<-\frac{1}{2}$
(Fig.~\ref{Fig:OperatingRegimeEta-1})} \label{App:A} In this case,
the inequalities $\gamma>\eta$ and $\beta < 1-2\eta$ always hold.
These inequalities are equivalent to $\beta+\gamma>\beta+\eta$ and
$\frac{1+\beta}{2}>\beta+\eta$, respectively. Hence, it follows that
$e_{\textrm{IMH}}=\beta+\eta$ for all the operating regimes, thus
resulting in
$\max\{e_{\textrm{ISH}},e_{\textrm{IMH}}\}=\max\left\{\min\left\{1+\gamma-\frac{\alpha(1-\beta)}{2},\beta+\eta\right\},\beta+\eta\right\}=\beta+\eta$.
Since $\beta+\eta<\max\{e_{\textrm{MH}},e_{\textrm{HC}}\}$ from
$\beta+\eta<e_{\textrm{MH}}$,  the entire operating regimes are
Regime A as illustrated in Fig.~\ref{Fig:OperatingRegimeEta-1}.

\subsubsection{$-\frac{1}{2}\leq\eta <0$
(Fig.~\ref{Fig:OperatingRegimeEta-2})} As in Appendix~\ref{App:A}
($\eta<-\frac{1}{2}$), we have
$\max\{e_{\textrm{ISH}},e_{\textrm{IMH}}\}=\beta+\eta$ from
$\gamma>\eta$ and $\beta < 1-2\eta$.

\begin{itemize}
\item $\beta<\frac{1}{2}-\eta$: In this case, from $\beta+\eta<e_{\textrm{MH}}$, it follows that $\beta+\eta<\max\{e_{\textrm{MH}},e_{\textrm{HC}}\}$, and thus the associated operating regimes belong to Regime A.

\item $\beta\geq\frac{1}{2}-\eta$: In this case, it follows that $\beta+\eta\geq e_{\textrm{MH}}$.
Hence, the throughput scaling exponent of the best achievable scheme
is given by
\begin{align}
&\max\left\{\beta+\eta,e_{\textrm{MH}},e_{\textrm{HC}}\right\}=\max\left\{\beta+\eta,e_{\textrm{HC}}\right\}
\nonumber\\& =\max\left\{\beta+\eta,2-\frac{\alpha}{2}\right\},
\nonumber
\end{align}
and the associated regimes belong to Regime $\tilde{\textrm{B}}$.
\end{itemize}
The operating regimes for $-\frac{1}{2}\leq\eta <0$ are finally
illustrated in Fig. \ref{Fig:OperatingRegimeEta-2}.

\subsubsection{$0\leq\eta <\frac{1}{2}$ \label{App:C}
(Fig.~\ref{Fig:OperatingRegimeEta-3})} If
$\beta+\gamma<\frac{1}{2}$, then pure ad hoc transmissions without
BSs outperform the infrastructure-supported protocols. Thus, the
associated operating regimes are included in Regime $\textrm{A}$.
Now, let us focus on the other case, i.e.,
$\beta+\gamma\ge\frac{1}{2}$, in the following.

\begin{itemize}
\item $\gamma>\eta$ and $\beta<1-2\eta$: When $\beta+2\gamma<1$, it follows that $e_{\textrm{IMH}}=\min\{\beta+\gamma,\beta+\eta\}=\beta+\eta$ from $\gamma>\eta$.
When $\beta+2\gamma\geq 1$, it follows that
$e_{\textrm{IMH}}=\min\left\{\frac{1+\beta}{2},\beta+\eta\right\}=\beta+\eta$
because $\beta<1-2\eta$. Hence, we have
$\max\{e_{\textrm{ISH}},e_{\textrm{IMH}}\}=\beta+\eta$. If
$\beta<\frac{1}{2}-\eta$ holds, then the associated operating
regimes belong to Regime $\textrm{A}$ from the fact that
$\beta+\eta<e_{\textrm{MH}}$. Otherwise, it follows that
$\beta+\eta\geq e_{\textrm{MH}}$, and thus the throughput scaling
exponent of the best achievable scheme is given by
$\max\left\{\beta+\eta,2-\frac{\alpha}{2}\right\}$. In other words,
for $\beta\geq\frac{1}{2}-\eta$, the network is
infrastructure-limited as $\alpha$ increases beyond a certain value,
and the resulting operating regimes belong to Regime
$\tilde{\textrm{B}}$.

\item $\gamma<\eta$ and $\beta+2\gamma<1$: The throughput scaling exponent achieved by the two
infrastructure-supported protocols is given by
\begin{align}
&\max\{e_{\textrm{ISH}},e_{\textrm{IMH}}\}
\nonumber\\
&=\min\left\{\max\left\{1+\gamma-\frac{\alpha(1-\beta)}{2},\beta+\gamma\right\},\beta+\eta\right\}
\nonumber\\ &=\min\{\beta+\gamma,\beta+\eta\}=\beta+\gamma.
\nonumber
\end{align}
Hence, the operating regimes belong to Regime B.

\item $\beta+2\gamma\geq 1$, $\beta\geq 1-2\eta$, and $\gamma\geq \beta^2+(\eta-2)\beta+1$:
In this case, note that the performance based on the IMH protocol is
not limited by backhaul transmission rate since $\beta+\eta\ge
\frac{1+\beta}{2}$, or equivalently, $\beta\geq 1-2\eta$. Now, let
us deal with the ISH protocol. When both conditions
$e_{\textrm{ISH}}=\beta+\eta$ and $\beta+\eta\geq e_{\textrm{HC}}$
follow, the network is infrastructure-limited for some $\alpha$,
where $\beta+\eta$ is always greater than or equal to
$e_{\textrm{MH}}$ (i.e., $\beta\ge \frac{1}{2}-\eta$) in this case.
More specifically, we need the condition $4-2\beta-2\eta\leq \alpha<
2+\frac{2(\gamma-\eta)}{1-\beta}$ so that the above two conditions
hold (see Table~\ref{Tab:RateEta} for the details). It thus follows
that $\gamma\geq \beta^2+(\eta-2)\beta+1$ from $4-2\beta-2\eta\leq
2+\frac{2(\gamma-\eta)}{1-\beta}$. In consequence, the associated
operating regimes belong to Regime $\tilde{\textrm{D}}$.

\item $\beta+2\gamma\geq 1$, $\beta\geq 1-2\eta$, and $\gamma< \beta^2+(\eta-2)\beta+1$:
In this case, we again note that the performance based on the IMH
protocol is not limited by backhaul transmission rate since
$\beta\geq 1-2\eta$. Now, let us focus on the ISH protocol at the
path-loss attenuation regime $\alpha<
2+\frac{2(\gamma-\eta)}{1-\beta}$ (see Table~\ref{Tab:RateEta}).
Then, the throughput scaling exponent $\beta+\eta$ achieved by the
ISH protocol is less than $e_{\textrm{HC}}=2-\frac{\alpha}{2}$ since
$\alpha<4-2\beta-2\eta$ due to the fact that
$2+\frac{2(\gamma-\eta)}{1-\beta} < 4-2\beta-2\eta $. The associated
operating regimes are thus the same as those for the network with
infinite-capacity infrastructure depicted in
Fig.~\ref{Fig:OperatingRegimesInfiniteBScapacity}.
\end{itemize}
Finally, the operating regimes for $0\leq\eta<\frac{1}{2}$ are
illustrated in Fig.~\ref{Fig:OperatingRegimeEta-3}.

\subsubsection{$\frac{1}{2}\leq\eta<1$
(Fig.~\ref{Fig:OperatingRegimeEta-4})} \label{App:D} When
$\beta+2\gamma<1$, it follows that $e_{\textrm{IMH}}=\beta+\gamma$,
which is less than $\beta+\eta$. When $\beta+2\gamma\geq 1$, we have
$e_{\textrm{IMH}}=\frac{1+\beta}{2}$, which is less than
$\beta+\eta$ since $\beta>1-2\eta$. For this reason, the performance
based on the IMH protocol is not limited by backhaul transmission
rate for all the operating regimes with respect to $\beta$ and
$\gamma$. For the ISH protocol, following the same line as the third
bullet of Appendix~\ref{App:C} ($0\leq\eta <\frac{1}{2}$), we obtain
the condition $\gamma\geq \beta^2+(\eta-2)\beta+1$, which is
depicted in Regime $\tilde{\textrm{D}}$. Other regimes are not
infrastructure-limited and is the same as those for the network with
infinite-capacity infrastructure. The operating regimes for
$\frac{1}{2}\leq\eta<1$ are finally illustrated in Fig.
\ref{Fig:OperatingRegimeEta-4}.

\subsubsection{$\eta\geq 1$
(Fig.~\ref{Fig:OperatingRegimesInfiniteBScapacity})} As in
Appendix~\ref{App:D} ($\frac{1}{2}\le \eta<1$), the network using
the IMH protocol is not limited by the backhaul transmission, i.e.,
$e_{\textrm{IMH}}<\beta+\eta$. Furthermore, in this case, there is
no operating regime such that both conditions $\gamma\geq
\beta^2+(\eta-2)\beta+1$ and $\beta+\gamma\leq 1$ hold. Hence, the
entire regimes are not infrastructure-limited, and thus the
resulting operating regimes are exactly the same as those for the
network with infinite-capacity infrastructure.

%%%%%%%%%%%%%%%%%%%%%%%%%%%%%%%%%%%%%%%%%%%%%%%%%%%%%%%%%%%%%%%%%%%%%%%%%%%%%%%%%%%%%%%%%%%%%%%%%%%%%%%%%%%%%%%%%%%%%%%%%%%%%%%%%%%%%%%%%%%%%%%%%%%%%


\begin{thebibliography}{10}
\providecommand{\url}[1]{#1}
\def\UrlFont{\rmfamily}
\providecommand{\newblock}{\relax} \providecommand{\bibinfo}[2]{#2}
\providecommand\BIBentrySTDinterwordspacing{\spaceskip=0pt\relax}
\providecommand\BIBentryALTinterwordstretchfactor{4}
\providecommand\BIBentryALTinterwordspacing{\spaceskip=\fontdimen2\font
plus \BIBentryALTinterwordstretchfactor\fontdimen3\font minus
  \fontdimen4\font\relax}
\providecommand\BIBforeignlanguage[2]{{%
\expandafter\ifx\csname l@#1\endcsname\relax
\typeout{** WARNING: IEEEtran.bst: No hyphenation pattern has been}%
\typeout{** loaded for the language `#1'. Using the pattern for}%
\typeout{** the default language instead.}%
\else \language=\csname l@#1\endcsname \fi #2}}

\bibitem{GuptaKumar:00}
P.~Gupta and P.~R. Kumar, ``The capacity of wireless networks,''
\emph{{IEEE}
  Trans. Inf. Theory}, vol.~46, no.~2, pp. 388--404, Mar. 2000.

\bibitem{D.Knuth:76}
D.~E. Knuth, ``Big omicron and big omega and big theta,'' \emph{ACM
SIGACT
  News}, vol.~8, no.~2, pp. 18--24, Apr.-Jun. 1976.

\bibitem{FranceschettiDouseTseThiran:07}
M.~Franceschetti, O.~Dousse, D.~N.~C. Tse, and P.~Thiran, ``Closing
the gap in
  the capacity of wireless networks via percolation theory,'' \emph{{IEEE}
  Trans. Inf. Theory}, vol.~53, no.~3, pp. 1009--1018, Mar. 2007.

\bibitem{ShinChungLee:TIT13}
W.-Y. Shin, S.-Y. Chung, and Y.~H. Lee, ``Parallel opportunistic
routing in
  wireless networks,'' \emph{{IEEE} Trans. Inf. Theory}, vol.~59, no.~10, pp.
  6290--6300, Oct. 2013.

\bibitem{ElGamalMammenPrabhakarShah:06}
A.~{El Gamal}, J.~Mammen, B.~Prabhakar, and D.~Shah, ``Optimal
throughput-delay
  scaling in wireless networks--{P}art {I}: {T}he fluid model,'' \emph{{IEEE}
  Trans. Inf. Theory}, vol.~52, no.~6, pp. 2568--2592, June 2006.

\bibitem{OzgurLevequeTse:07}
A.~{\"O}zg{\"u}r, O.~L{\'e}v{\^e}que, and D.~N.~C. Tse,
``Hierarchical
  cooperation achieves optimal capacity scaling in ad hoc networks,''
  \emph{{IEEE} Trans. Inf. Theory}, vol.~53, no.~10, pp. 3549--3572, Oct. 2007.

\bibitem{NiesenGuptaShah:10}
U.~Niesen, P.~Gupta, and D.~Shah, ``The balanced unicast and
multicast capacity
  regions of large wireless networks,'' \emph{{IEEE} Trans. Inf. Theory},
  vol.~56, no.~5, pp. 2249--2271, May 2010.

\bibitem{GrossglauserTse:02}
M.~Grossglauser and D.~N.~C. Tse, ``Mobility increases the capacity
of ad hoc
  wireless networks,'' \emph{{IEEE/ACM} Trans. Networking}, vol.~10, no.~4, pp.
  477--486, Aug. 2002.

\bibitem{CadambeJafar:08}
V.~R. Cadambe and S.~A. Jafar, ``Interference alignment and degrees
of freedom
  of the {$K$}-user interference channel,'' \emph{{IEEE} Trans. Inf. Theory},
  vol.~54, no.~8, pp. 3425--3441, Aug. 2008.

\bibitem{Niesen:IT11}
U.~Niesen, ``Interference alignment in dense wireless networks,''
\emph{{IEEE}
  Trans. Inf. Theory}, vol.~57, no.~5, pp. 2889--2901, May 2011.

\bibitem{LiZhangFang:TMC11}
P.~Li, C.~Zhang, and Y.~Fang, ``The capacity of wireless ad hoc
networks using
  directional antennas,'' vol.~10, no.~10, pp. 1374--1387, Oct. 2011.

\bibitem{ZemlianovVeciana:05}
A.~Zemlianov and G.~de~Veciana, ``Capacity of ad hoc wireless
networks with
  infrastructure support,'' \emph{{IEEE} J. Select. Areas Commun.}, vol.~23,
  no.~3, pp. 657--667, Mar. 2005.

\bibitem{GarettoGiacconeLeonardi:TN09-1}
M.~Garetto, P.~Giaccone, and E.~Leonardi, ``Capacity scaling in ad
hoc networks
  with heterogeneous mobile nodes: The super-critical regime,''
  \emph{{IEEE/ACM} Trans. Networking}, vol.~17, no.~5, pp. 1522--1535, Oct.
  2009.

\bibitem{GarettoGiacconeLeonardi:TN09-2}
------, ``Capacity scaling in ad hoc networks with heterogeneous mobile nodes:
  The subcritical regime,'' \emph{{IEEE/ACM} Trans. Networking}, vol.~17,
  no.~6, pp. 1888--1901, Dec. 2009.

\bibitem{AlfanoGarettoLeonardiMartina:TN10}
G.~Alfano, M.~Garetto, E.~Leonardi, and V.~Martina, ``Capacity
scaling of
  wireless networks with inhomogeneous node density: Lower bounds,''
  \emph{{IEEE/ACM} Trans. Networking}, vol.~18, no.~5, pp. 1624--1636, Oct.
  2010.

\bibitem{YinGaoCui:TN10}
C.~Yin, L.~Gao, and S.~Cui, ``Scaling laws for overlaid wireless
networks: A
  cognitive radio network versus a primary network,'' \emph{{IEEE/ACM} Trans.
  Networking}, vol.~18, no.~4, pp. 1317--1329, Aug. 2010.

\bibitem{HuangWang:TN12}
W.~Huang and X.~Wang, ``Capacity scaling of general cognitive
networks,''
  \emph{{IEEE/ACM} Trans. Networking}, vol.~20, no.~5, pp. 1501--1513, Oct.
  2012.

\bibitem{O.Dousse:INFOCOM02}
O.~Dousse, P.~Thiran, and M.~Hasler, ``Connectivity in ad-hoc and
hybrid
  networks,'' in \emph{Proc. {IEEE} {INFOCOM}}, New York, NY, June 2002, pp.
  1079--1088.

\bibitem{KozatTassiulas:03}
U.~C. Kozat and L.~Tassiulas, ``Throughput capacity of random ad hoc
networks
  with infrastructure support,'' in \emph{Proc. {ACM} {MobiCom}}, San Diego,
  CA, Sept. 2003, pp. 55--65.

\bibitem{LiuThiranTowsley:07}
B.~Liu, P.~Thiran, and D.~Towsley, ``Capacity of a wireless ad hoc
network with
  infrastructure,'' in \emph{Proc. {ACM} {M}obi{H}oc}, Montr{\'e}al, Canada,
  Sept. 2007, pp. 239--246.

\bibitem{ShinJeonDevroyeVuChungLeeTarokh:08}
W.-Y. Shin, S.-W. Jeon, N.~Devroye, M.~H. Vu, S.-Y. Chung, Y.~H.
Lee, and
  V.~Tarokh, ``Improved capacity scaling in wireless networks with
  infrastructure,'' \emph{{IEEE} Trans. Inf. Theory}, vol.~57, no.~8, pp.
  5088--5102, Aug. 2011.

\bibitem{GuthyUtschickHonig:JSAC13}
C.~Guthy, W.~Utschick, and M.~L. Honig, ``Large system analysis of
sum capacity
  in the {G}aussian {MIMO} broadcast channel,'' \emph{{IEEE} J. Select. Areas
  Commun.}, vol.~31, no.~2, pp. 149--159, Feb. 2013.

\bibitem{C.Capar:11}
\c{C}. \c{C}apar, D.~Goeckel, D.~Towsley, R.~Gibbens, and A.~Swami,
``Cut
  results for the capacity of hybrid networks,'' in \emph{Proc. Annual Conf.
  Int. Technol. Alliance {(ACITA)}}, Adelphi, MD, Sept. 2011, pp. 1--2.

\bibitem{C.Capar:12}
------, ``Capacity of hybrid networks,'' in \emph{Proc. Annual Conf. Int.
  Technol. Alliance {(ACITA)}}, Southampton, UK, Sept. 2012, pp. 1--8.

\bibitem{JeongShin:ISIT13}
C.~Jeong and W.-Y. Shin, ``Large-scale ad hoc networks with
rate-limited
  infrastructure: Information-theoretic operating regimes,'' in \emph{Proc.
  IEEE Int. Symp. Inf. Theory {(ISIT)}}, Istanbul, Turkey, July 2013, pp.
  424--428.

\bibitem{S.Sesia:11}
S.~Sesia, I.~Toufik, and M.~Baker, \emph{LTE - The UMTS Long Term
Evolution:
  From Theory to Practice}.\hskip 1em plus 0.5em minus 0.4em\relax UK: Wiley,
  2011.

\bibitem{A.Sanderovich:ISIT07}
A.~Sanderovich, O.~Somekh, and S.~Shamai, ``Uplink macro diversity
with limited
  backhaul capacity,'' in \emph{Proc. ISIT}, Nice, France, June 2007, pp.
  11--15.

\bibitem{P.Marsch:ICC07}
P.~Marsch and G.~Fettweis, ``A framework for optimizing the uplink
performance
  of distributed antenna systems under a constrained backhaul,'' in \emph{Proc.
  IEEE Int. Conf. Communications (ICC'07)}, Glasgow, Scotland, June 2007, pp.
  975--979.

\bibitem{S.Shamai:PIMRC08}
S.~Shamai, O.~Simeone, O.~Somekh, and A.~Sanderovich,
``Information-theoretic
  implications of constrained cooperation in simple cellular models,'' in
  \emph{Proc. PIMRC}, Cannes, France, Sept. 2008, pp. 1--5.

\bibitem{B.Nazer:ISIT09}
B.~Nazer, A.~Sanderovich, M.~Gastpar, and S.~Shamai, ``Structured
superposition
  for backhaul constrained cellular uplink,'' in \emph{Proc. IEEE Int. Symp.
  Information Theory (ISIT'09)}, Seoul, Korea, June-July 2009, pp. 1530--1534.

\bibitem{A.Sanderovich:TIT09}
A.~Sanderovich, O.~Somekh, H.~V. Poor, and S.~Shamai, ``Uplink macro
diversity
  of limited backhaul cellular network,'' \emph{{IEEE} Trans. Inf. Theory},
  vol.~55, no.~8, pp. 3457--3478, Aug. 2009.

\bibitem{ParkSimeoneSahinShamai:TSP13}
S.-H. Park, O.~Simeone, O.~Sahin, and S.~Shamai, ``Joint precoding
and
  multivariate backhaul compression for the downlink of cloud radio access
  networks,'' \emph{{IEEE} Trans. Signal Processing}, vol.~61, no.~22, pp.
  5646--5658, Nov. 2013.

\bibitem{OzgurJohariTseLeveque:10}
A.~{\"O}zg{\"u}r, R.~Johari, D.~N.~C. Tse, and O.~L{\'e}v{\^e}que,
  ``Information-theoretic operating regimes of large wireless networks,''
  \emph{{IEEE} Trans. Inf. Theory}, vol.~56, no.~1, pp. 427--437, Jan. 2010.

\bibitem{ViswanathTse:03}
P.~Viswanath and D.~N.~C. Tse, ``Sum capacity of the vector
{G}aussian
  broadcast channel and uplink-downlk duality,'' \emph{{IEEE} Trans. Inf.
  Theory}, vol.~49, pp. 1912--1921, Aug. 2003.

\bibitem{GomezRanganErkip:ISIT14}
F.~Gomez-Cuba, S.~Rangan, and E.~Erkip, ``Scaling laws for
infrastructure
  single and multihop wireless networks in wideband regimes,'' in \emph{Proc.
  IEEE Int. Symp. Inf. Theory (ISIT)}, Honolulu, HI, Jun./Jul. 2014, pp.
  76--80.

\bibitem{ConstantinescuScharf:98}
F.~Constantinescu and G.~Scharf, ``Generalized {G}ram-{H}adamard
inequality,''
  \emph{Journal of Inequalities and Applications}, vol.~2, pp. 381--386, 1998.

\bibitem{JovicicViswanathKulkarni:04}
A.~Jovicic, P.~Viswanath, and S.~R. Kulkarni, ``Upper bounds to
transport
  capacity of wireless networks,'' \emph{{IEEE} Trans. Inf. Theory}, vol.~50,
  no.~11, pp. 2555--2565, Nov. 2004.

\end{thebibliography}
\end{document}